\documentclass[10pt,journal,compsoc]{IEEEtran}
\renewcommand{\baselinestretch}{0.98}
\ifCLASSOPTIONcompsoc
  \usepackage{cite}
\else
  \usepackage{cite}
\fi

\usepackage{graphicx}
\usepackage{comment}
\usepackage{amsmath,amssymb, bm, amsfonts} 
\usepackage{mathrsfs}
\usepackage{color}
\usepackage{colortbl}  
\usepackage{xcolor}
\usepackage{ragged2e}

\usepackage{multirow}
\usepackage{multicol}
\usepackage{floatrow}
\usepackage[noend]{algpseudocode}
\usepackage{algorithmicx,algorithm}
\usepackage{bbding}
\usepackage{stfloats} 

\usepackage{amsthm}   

\newtheorem{theorem}{Theorem}
\newtheorem{lemma}{Lemma}
\newtheorem{definition}{Definition}
\newtheorem{assumption}{Assumption}

\newtheorem{remark}{Remark}

\algrenewcommand\algorithmicrequire{\textbf{Input:}}
\algrenewcommand\algorithmicensure{\textbf{Output:}}

\usepackage{tikz,xcolor,hyperref}
\usepackage{subfigure}
\usepackage{booktabs}
\usepackage{makecell}
\usepackage{xspace}
\usepackage{threeparttable}
\usepackage{pifont}
\newcommand{\cmark}{\ding{51}}
\newcommand{\xmark}{\ding{55}}
\colorlet{goodbg}{green!18}
\colorlet{okbg}{red!10}

\newcommand{\pref}[1]{\textbf{#1}~\cmark}
\newcommand{\notpref}[1]{#1~\xmark}

\definecolor{lime}{HTML}{A6CE39}
\DeclareRobustCommand{\orcidicon}{%
    \begin{tikzpicture}
    \draw[lime, fill=lime] (0,0) 
    circle [radius=0.16] 
    node[white] {{\fontfamily{qag}\selectfont \tiny ID}};    \draw[white, fill=white] (-0.0625,0.095) 
    circle [radius=0.007];    \end{tikzpicture}
    \hspace{-2mm}}
\foreach \x in {A, ..., Z}{%
    \expandafter\xdef\csname orcid\x\endcsname{\noexpand\href{https://orcid.org/\csname orcidauthor\x\endcsname}{\noexpand\orcidicon}}
    }

\begin{document}
%
\title{Decentralized GNSS at Global Scale via Graph-Aware Diffusion Adaptation}
%
%
%
%

\author{Xue Xian Zheng$^{1}$, Xing Liu$^{2,3,\dagger}$, and Tareq Y. Al-Naffouri$^{1}$ \\
$^{1}$ King Abdullah University of Science and Technology, Saudi Arabia\\
$^{2}$ Institute for Space Studies of Catalonia, Spain \\
$^{3}$ Universitat Autònoma de Barcelona, Spain
\thanks{
$^\dagger$Corresponding author. e-mail: xing@ieec.cat
}
}
%
%

\markboth{}%
%
\\
\IEEEtitleabstractindextext{%
\begin{abstract}\justifying
Network-based Global Navigation Satellite Systems (GNSS) support critical infrastructure and autonomous systems, but centralized processing hubs limit scalability, resilience, and latency. We present a global-scale decentralized GNSS architecture that jointly estimates receiver states and produces network-wide satellite-correction products. Modeling the receiver network as a time-varying graph, we formulate information-mixing schedule design as training a deep linear neural network whose layers are doubly stochastic matrices constrained by graph topology. Masked Sinkhorn–Knopp projections are integrated with backpropagation to preserve topology-awareness and feasibility during training. The learned schedule supports an online gradient-tracking diffusion strategy, allowing receivers to perform local inference from their own observations while exchanging compact messages to reach consensus on satellite corrections and self-localization. When receivers are reference stations, the resulting consensus products can be broadcast for precise point positioning (PPP) and precise point positioning--real-time kinematic (PPP--RTK) services. Experiments using hundreds of globally distributed IGS stations show that the proposed method matches centralized baselines while converging faster and reducing communication overhead relative to existing decentralized approaches. Overall, by reframing decentralized GNSS as a networked signal processing problem, this work highlights the potential of decentralized optimization, consensus-based inference, and graph-aware learning as effective tools for operational satellite navigation.
\end{abstract}

}

\maketitle

\IEEEdisplaynontitleabstractindextext

%
\IEEEpeerreviewmaketitle

\section*{Introduction}
High-precision Global Navigation Satellite Systems (GNSS) provide the foundational layer for applications ranging from autonomous transport to crustal deformation monitoring. To achieve centimeter-level precision—such as in precise point positioning (PPP) ~\cite{heroux1995gps,zumberge1997precise} and precise point positioning–real-time kinematic (PPP-RTK)~\cite{wabbena2005ppp} modes—user receivers rely on network-derived products, including precise satellite clocks, orbits, and atmospheric corrections.  Historically, these products are estimated in a centralized manner: observations from global networks of continuously operating reference stations (CORS) are streamed to a single analysis center for joint processing \cite{teunissen2010ppp,laurichesse2009integer,collins2010undifferenced}. While effective at current scales, this centralized paradigm is approaching a critical threshold, strained by an exponential increase in data volume and network density.

The rapid expansion of multi-constellation tracking and the imminent integration of Low Earth Orbit (LEO) augmented GNSS are creating a big-data regime characterized by high observation densities and short update intervals~\cite{10139999,10542356}. Proposed LEO constellations alone are projected to comprise over $60$k satellites over the coming years, according to a U.S. Government Accountability Office technology report~\cite{gao2022largeconstellations}. Although only a subset of these LEO constellations is explicitly designed for navigation, an increasing number of non-dedicated systems are being actively explored by the GNSS community as signals of opportunity for positioning, thereby potentially enabling an unprecedented number of spaceborne transmitters to be leveraged for navigation~\cite{10140066}. In this context, centralized processing creates severe vulnerabilities, including prohibitive computational complexity, bandwidth saturation from raw data backhaul, and latency in product dissemination. Moreover, reliance on a monolithic analysis hub introduces a resilience bottleneck, where node failures or link congestion can compromise the entire service. To support the next generation of instantaneous, ubiquitous positioning, the computational burden must shift from the center to the edge.

Decentralized architectures~\cite{sayed2014diffusion,4118472} offer a compelling response to these systemic vulnerabilities, as evidenced by the recent success of distributed systems in blockchain networks~\cite{wood2014ethereum}, large-scale training on neural models~\cite{mcmahan2017communication}, and robotic swarm intelligence~\cite{rubenstein2014programmable}. In these paradigms, computation and storage are moved away from a single monolithic server and distributed across many loosely coupled nodes, which exchange only compact, high-level information (e.g., gradients, model updates, or consensus variables) rather than raw data streams~\cite{4118472}. This shift enables systems that scale with the number of participants, improve robustness through redundancy, and adapt locally to heterogeneous conditions, while maintaining low end-to-end latency~\cite{mcmahan2017communication}.

Translating these ideas to GNSS networks is particularly promising. Decentralized processing can alleviate central bottlenecks and enable new classes of real-time services. Edge nodes—such as reference stations or high-end user receivers—can perform local preprocessing, quality control, and partial estimation of satellite and atmospheric parameters, while exchanging only compact correction updates or state summaries with neighboring nodes and a lightweight coordination layer. However, the realization
of truly decentralized high-precision GNSS is still at an
early stage. Yet the idea of distributing the processing
workload across the network traces back more than a
decade. Boomkamp's DANCER system was an early attempt to decentralize large-scale geodetic GNSS batch processing by eliminating receiver-specific unknowns locally and transmitting only reduced information on shared global parameters, thereby avoiding the transfer of full normal-equation systems~\cite{boomkamp2012distributed}. However, its accuracy depends on the statistical assumption that certain local cross-covariances are negligible, which is mainly justified for sufficiently large and diverse datasets. Khodabandeh, Teunissen, and Zaminpardaz later formulated GNSS processing as a decentralized recursive filtering problem rather than a batch normal-equation adjustment~\cite{khodabandeh2017consensus}. Their consensus-based Kalman filter estimates a common time-varying state observed by multiple receivers, with the state model describing shared GNSS parameters such as ionospheric quantities and instrumental biases, rather than receiver-specific motion. Decentralization is enabled by the additive information-form update, where local information matrices and vectors are combined through neighbor-based average consensus to approximate the centralized global update. Although the framework is general, its GNSS demonstration is limited to a small regional real-time filtering example and does not establish scalability to global, multi-parameter GNSS service generation. A closer methodological predecessor is the decentralized alternating direction method of multipliers (ADMM) network of Khodabandeh and Teunissen, which estimates satellite parameters over a global receiver network without relying on the statistical approximation used in DANCER~\cite{khodabandeh2019distributed}. By exchanging primal and dual variables among neighboring nodes, their method achieves convergence to centralized satellite code-bias solutions, but it is mainly shown for static and well-connected communication graphs. More recently, Hou and Zhang proposed a semi-decentralized PPP--RTK product-generation strategy based on subnetwork processing, datum alignment, and atmospheric refinement, achieving centralized-level accuracy with substantially reduced computational cost~\cite{hou2023decentralized}. Overall, existing studies mark important steps toward distributed GNSS architectures, but they remain limited by statistical assumptions, communication and convergence requirements, partial decentralization, or insufficient exploitation of graph structure.

In this paper, we reformulate decentralized GNSS processing as a problem of networked signal and information processing (SIP)~\cite{Multitaskg,vlaski2023networked,zheng2025error,zheng2026quantitative}, enabling joint receiver-state estimation and network-wide satellite-correction generation directly at the edge. Unlike conventional schemes that rely on fixed or limited communication patterns, the proposed framework explicitly supports time-varying receiver networks and global-scale deployment. By exploiting the underlying graph structure, the framework provides a principled approach to improving scalability, communication efficiency, and latency, thereby enabling scalable decentralized GNSS estimation and correction architectures. 

Our main contributions are as follows: 
 \begin{itemize}
    \item  We formalize the global GNSS receiver network as a time-varying graph, where nodes represent receivers and edges capture time-varying communication links (\textbf{Figure}~\ref{fig:igs-graph}). We then formulate the design of inter-receiver information mixing schedules as the offline training of a deep linear neural network. Each layer is constrained to be doubly stochastic with a sparsity pattern consistent with the underlying graph topology, and the finite product of these layers is optimized to approximate the consensus operator. Feasibility is ensured during training via masked Sinkhorn--Knopp projections, which enforce both topology awareness and doubly stochasticity (\textbf{Figure}~\ref{fig_1}, \textbf{Left Panel}; \textbf{Figure}~\ref{fig:training}). 
    \item We develop a fully decentralized online diffusion strategy that applies the learned mixing schedules to process GNSS observations. By representing local receiver states and global satellite parameters as latent graph signals, each receiver iteratively approaches the centralized solution using only local observations and compact message exchanges with neighboring nodes (\textbf{Figure}~\ref{fig_1}, \textbf{Right Panel}). As a complementary theoretical result, we also derive convergence guarantees and establish linear convergence rates under the stated assumptions. 
    \item We evaluate the proposed architecture on a global GNSS network constructed from the geodetic coordinates of hundreds of IGS stations and corresponding RINEX navigation data. Numerical results show that the decentralized estimates match the centralized baseline, while converging faster than existing decentralized methods and reducing communication overhead (\textbf{Figure}~\ref{fig:gt-vs-centralized}; \textbf{Figure}~\ref{fig:consensus-metropolis}). These findings suggest a viable way to address the big-data challenges of large-scale satellite navigation. 
\end{itemize}

\begin{figure*}[!t]
\centering
\includegraphics[width=0.85\linewidth,height=8.5cm]{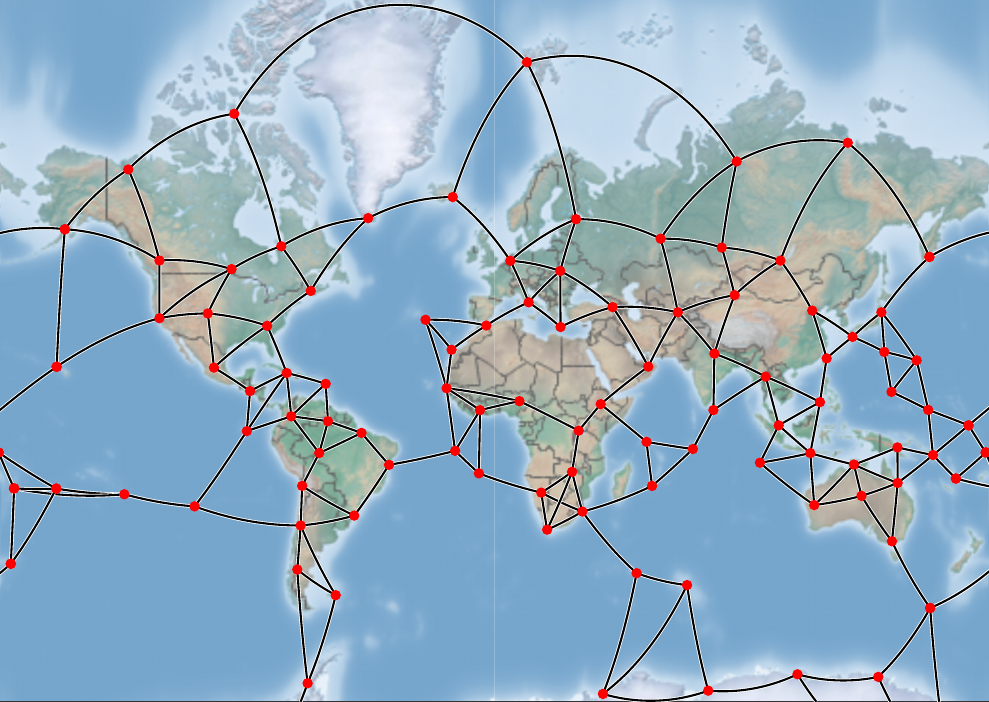}
\caption{Decentralized GNSS receiver network used in this work: a sparse, global graph with $100$ nodes placed at selected IGS station locations. Red points denote nodes; black lines denote communication links. The depicted communication pattern yields a graph characterized by $182$ edges, a diameter of $12$, and a radius of $10$. Note that this represents a single instance; patterns can change under different protocols, and the associated mixing weights that govern information exchange between nodes also vary; together, these constitute a time-varying communication graph.
}
\label{fig:igs-graph}
\end{figure*}
The remainder of this paper is organized as follows. The \textbf{Results} section presents the overall design of the proposed architecture and evaluates its performance across a range of metrics. The \textbf{Discussion} section focuses on the limitations of the current architecture and possible directions for future research. The \textbf{Methods} section describes the detailed GNSS observation model, its transformation into a decentralized formulation, the offline learning of mixing schedules, and the online decentralized diffusion strategy. The theoretical convergence analysis is provided in the \textbf{Supplementary Information}.

\begin{figure*}[!t]
\centering
\includegraphics[width=1\textwidth]{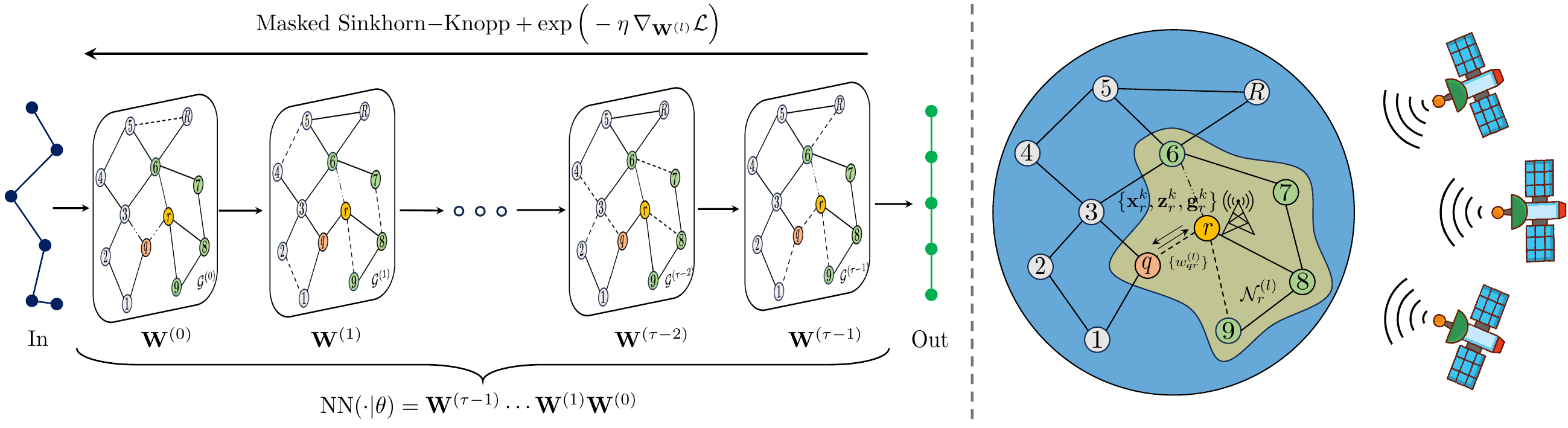}
\caption{{\bf{(Left.)}} Learning graph-aware mixing sequence $\theta=\{\mathbf{W}^{(l)}\}_{l=0}^{\tau-1}$ on time-varying graph $\mathcal{G}^{(l)}=(\mathcal{V},\mathcal{E}^{(l)},\mathbf{W}^{(l)})$. We regard the product \(\mathrm{NN}(\cdot|\theta)= \mathbf{W}^{(\tau-1)}\cdots \mathbf{W}^{(0)}\) as a deep linear neural network and optimize \(\theta\) via exponentiated-gradient updates, masking, and Sinkhorn Knopp-projection. This enforces each \(\mathbf{W}^{(l)}\) to be doubly stochastic while driving the product toward the consensus operator \(\tfrac{1}{R}\mathbf{1}\mathbf{1}^\top\) with small error \(\epsilon_\tau\). {\bf{(Right.)}} ATC GT-Diffusion scheme with $\theta$ for GNSS network. Each receiver $r\in\mathcal{V}$ updates its state $(\mathbf{x}_r^{k}, \mathbf{z}_r^{k})$ and GT variable $\mathbf{g}_r^{k}$ at iteration $k$.
Updates use (i) local observations from the visible satellites and (ii) messages from neighbors $\mathcal{N}^{(l)}_{r}$, combined via mixing weight $w^{(l)}_{qr}$ for $q\in\mathcal{N}^{(l)}_{r}$ over a time-varying graph
$\mathcal{G}^{(l)}$. Here $l=k \bmod \tau$, i.e., $\mathbf{W}^{(l)}$ is sampled $\tau$-cyclically from $\theta$ across $k$, which promotes convergence.}
\label{fig_1}
\end{figure*}

\section*{Results}\label{sec:formulation}
\subsection*{The Overall Design}
The proposed global-scale, decentralized GNSS architecture operates as a two-stage system, as illustrated in Figure~\ref{fig_1}. The first stage is an offline graph-aware learning phase that derives a sequence of doubly stochastic mixing matrices tailored to a set of communication patterns. The second stage is an online optimization phase that executes an adapt-then-combine (ATC) diffusion strategy augmented with gradient tracking (GT). In this phase, nodes exchange global satellite parameters and gradient trackers to reach consensus on global states while simultaneously optimizing local parameters. These two stages are tightly coupled in a unified framework termed \textbf{\texttt{GraT-Diff}} (pronounced ``Great-Diff''), short for 
\textbf{Gra}ph-aware gradient-\textbf{T}racking \textbf{Diff}usion with an ATC protocol. The offline phase overcomes the limitations of single-snapshot topology optimization by maximizing information mixing over time, thereby significantly accelerating the online phase's convergence to the centralized baseline.

\subsection*{Agreement on Satellite Parameters}
We partition the GNSS parameter space at receiver $r$ into two distinct components: a local state $\mathbf{x}_r$ (capturing unique parameters such as incremental position, receiver clock offsets, atmospheric delays, and carrier-phase ambiguities) and a shared global state $\mathbf{z}$ (encompassing satellite clocks and hardware biases). The incremental position states follow the standard GNSS observed-minus-computed formulation, where the observation equations are linearized around a coarse receiver-state estimate. The resulting linearized GNSS observation model can be written as:
\begin{equation}
\mathbf{y}_r = \mathbf{A}_r \mathbf{x}_r + \mathbf{B}_r \mathbf{z} + \mathbf{n}_r,
\label{eq:obs_r}
\end{equation}
where $\mathbf{y}_r$ denotes the observed-minus-computed observation vector, $\mathbf{A}_r$ and $\mathbf{B}_r$ are the design matrices, and $\mathbf{n}_r$ represents measurement noise.

Because $\mathbf{z}$ represents the physical properties of the satellite common to all observers, it mathematically couples the estimation tasks throughout the network. Standard centralized approaches resolve this coupling by aggregating high-dimensional raw data at a single fusion center—a process that scales poorly. By contrast, we demonstrate that our framework enables nodes to reach strict consensus on the global variable $\mathbf{z}$ solely through compact peer-to-peer exchanges. This allows the network to recover the precise centralized solution while minimizing local objectives and maintaining communication efficiency (see \textbf{Methods} for full formulation details).

\subsection*{Communication graphs, Patterns and Graph-Aware Fast Mixing Weights}
As seen in Figure~\ref{fig:igs-graph}, we model the GNSS infrastructure not as a static hierarchy, but as a time-varying communication graph $\mathcal{G}^{(l)}=(\mathcal{V},\mathcal{E}^{(l)},\mathbf{W}^{(l)})$. Here, the $R$ receivers form the node set $\mathcal{V}=\{1,\ldots,R\}$, while the edge set $\mathcal{E}^{(l)} \subseteq \mathcal{V} \times \mathcal{V}$ represents time-varying communication links governed by bandwidth constraints and exchange protocols. To orchestrate information flow, each snapshot is assigned a mixing matrix $\mathbf{W}^{(l)}=[w^{(l)}_{rq}]\in\mathbb{R}^{R\times R} $, compatible with the topology,
where $w^{(l)}_{rq}>0$ only if receivers $r$ and $q$ can communicate. Furthermore, $\mathbf{W}^{(l)}$ is constrained to be doubly stochastic, satisfying $\mathbf{W}^{(l)}\mathbf{1}=\mathbf{1},\mathbf{1}^{\top}\mathbf{W}^{(l)}=\mathbf{1}^{\top}$, which ensures balanced information aggregation.

A fundamental challenge in such global-scale networks is that the graph may be sparsely connected, making the propagation of dispersed information across the network costly. To overcome this, our framework leverages the \textbf{block contraction property} that optimizes $\theta=\{\mathbf{W}^{(l)}\}_{l=0}^{\tau-1}$ over a time window $\tau\in\mathbb{N}_{\ge1}$. This relaxes the requirement of instantaneous connectivity and instead promotes effective information mixing over multiple graph snapshots, thus mitigating the communication burden across the network. Specifically, we recast the sequence of communication graphs as layers in a deep linear neural network $\mathrm{NN}(\cdot|\theta)$ (see Figure~\ref{fig_1}) and optimize the mixing weights $\theta$ such that their cumulative product approximates the ideal global averaging operator:
\begin{equation}
  \big\|\,\mathbf{W}^{(\tau-1)}\cdots \mathbf{W}^{(1)}\mathbf{W}^{(0)}
  - \tfrac{1}{R}\,\mathbf{1}\mathbf{1}^{\!\top}\big\| \;=\; \epsilon_\tau.
 \label{eq:prod}
\end{equation}
Here $\epsilon_\tau < 1$ is the spectral contraction factor over the window, with smaller values implying stronger contraction and faster consensus for fixed $\tau$. Specifically, when this block contraction is combined with the observation of homogeneous data across nodes, $\epsilon_\tau$ can be driven arbitrarily close to zero, provided $\tau$ is chosen between the graph diameter and roughly twice the graph radius, as suggested by previous work~\cite{hendrickx2014graph,fainman2024learned}. Several graph families even admit exact $\tau$-step averaging with $\epsilon_\tau = 0$, yielding \emph{finite-time consensus} with closed-form communication patterns and weights~\cite{ying2021exponential,nguyen2025graphs,liang2025understanding}, which significantly reduces communication overhead while accelerating convergence for decentralized methods. While these successes highlight a compelling opportunity for the co-optimization of physical GNSS receiver placement and communication patterns, existing infrastructure is calling for solutions that remain effective under prescribed topologies.
\begin{figure}[t]
  \centering
  \includegraphics[width=0.8\linewidth]{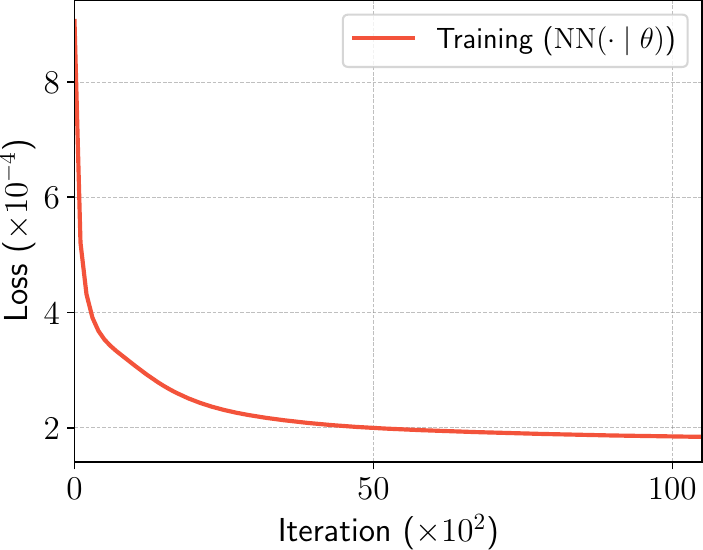}
    \caption{Training \(\mathrm{NN}(\cdot|\theta)\) via Algorithm~1 for the first $10^4$ iterations. We employ the ``Lazy'' parameterization $\frac{\mathbf{W}^{(l)}+\mathbf{I}}{2}$ at each layer to stabilize the subsequent online diffusion stage; this counteracts the effect of substantial data heterogeneity arising from potentially distinct satellite subsets observed by each receiver.}
  \label{fig:training}
\end{figure}

For the global IGS network, following the sparse communication pattern illustrated in Figure~\ref{fig:igs-graph}, the graph diameter is 12 and the radius is 10. Choosing $\tau$ below the diameter generally prevents information from propagating across the entire graph within a single communication block, whereas choosing $\tau$ significantly larger than twice the radius increases the number of trainable matrices and the memory required to store the schedule, with diminishing practical benefit. We therefore select $\tau=15$, which lies between the graph diameter and approximately twice the graph radius, consistent with guidance from \emph{finite-time consensus} results. In general, $\tau$ should scale with the mixing time or diameter of the effective communication graph: well-connected or expander-like graphs admit smaller $\tau$, whereas sparse, chain-like global graphs require larger windows. To additionally account for data heterogeneity arising from limited satellite visibility at each node, we reparameterize each layer in $\mathrm{NN}(\cdot|\theta)$ as $\frac{\mathbf{W}^{(l)}+\mathbf{I}}{2}$ (see \textbf{Methods} for details on this reparameterization trick). Under this configuration, training converges to a final loss of $1.77\times10^{-4}$ (Figure~\ref{fig:training}), corresponding to $\epsilon_\tau = 0.68$.

\subsection*{GraT-Diff Recovers the Centralized Solution}

We validate the proposed GNSS network using simulated measurements from $105$ satellites across the GPS, GLONASS, Galileo, and BeiDou constellations, derived from RINEX navigation files. Our method, \textbf{\texttt{GraT-Diff}}, integrates a standard online GT algorithm (see \textbf{Methods} for details on the algorithm) with learned, graph-aware mixing weights $\theta$.

We define the local state vector \(\mathbf{x}_r\) to encompass the receiver's incremental position, clock bias, and float ambiguity terms, leaving integer ambiguity resolution for a subsequent processing stage. The global consensus variable $\mathbf{z}$ captures the satellite-specific biases, with the first satellite fixed as a reference (bias set to zero) to ensure identifiability of the problem. Performance was benchmarked against the centralized Best Linear Unbiased Estimator (C-BLUE)~\cite{odijk2016estimability} utilizing two standard metrics relative to ground truth: the average relative positioning error $e_{\mathrm{pos}}$ and the mean absolute satellite bias error $e_{\mathbf{z}}$. We initialize all satellite-bias estimates, except for the reference satellite, to zero, corresponding to a stringent cold-start scenario in which the network has no prior information. In practical deployments, GNSS estimation is typically run recursively and uses the previous epoch’s solution as a prior, placing the initialization close to the optimum and leading to substantially faster convergence than in the pessimistic setting studied here.
\begin{figure}[t]
  \centering
  \includegraphics[width=0.8\linewidth]{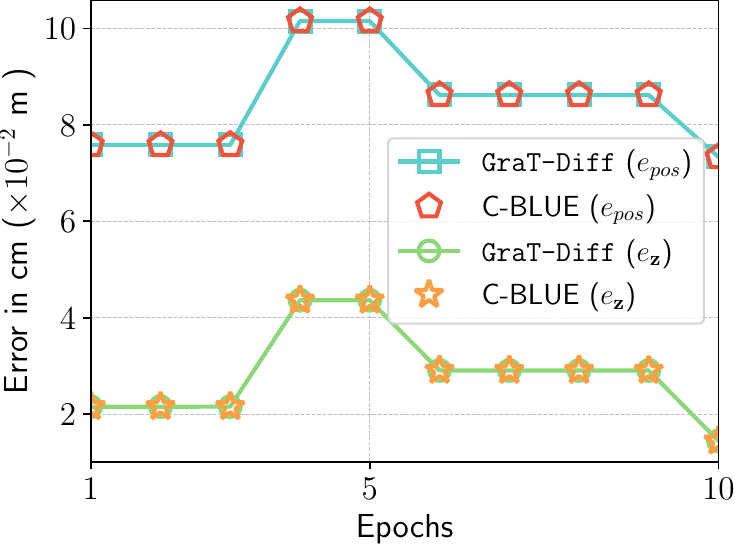}
    \caption{Accuracy evaluation between \textbf{\texttt{GraT-Diff}} and centralized solution (C-BLUE), exemplified by 10 epochs, containing local receiver's positioning error $e_{\mathrm{pos}}$ and global satellite-bias error $e_{\mathbf{z}}$. We report $e_{\mathbf{z}}$ in light-travel distance (equivalent range) units for consistency. Here, the maximum difference of $e_{\mathrm{pos}}$ and $e_{\mathbf{z}}$ between  \textbf{\texttt{GraT-Diff}} and C-BLUE are $6.88\times10^{-9}$m and $5\times10^{-9}$m, respectively, occurring at epoch 5.}
  \label{fig:gt-vs-centralized}
\end{figure}

\begin{table*}[!t]
\caption{Mahalanobis distance $d_M(\mathbf{z})$ corresponding to Fig.~\ref{fig:gt-vs-centralized}}\label{tab1}
\centering
\resizebox{\textwidth}{!}{%
\begin{tabular}{lcccccccccc}
\toprule
Epoch & 1 & 2  &  3  & 4  & 5  & 6  & 7  & 8  & 9  & 10 \\
\midrule
C-BLUE~\cite{odijk2016estimability}   & $11.1226...$  & $11.1226...$ & $11.1226...$ & $9.3236...$ & $9.3236...$ & $8.6940...$ & $8.6940...$ & $8.6940...$ & $8.6940...$ & $9.0973...$\\
Our Method & $11.1226...$  & $11.1226...$ & $11.1226...$ & $9.3236...$ & $9.3236...$ & $8.6940...$ & $8.6940...$ & $8.6940...$ & $8.6940...$ & $9.0973...$\\
Difference & $7.82\times10^{-9}$  & $-1.92\times10^{-8}$ & $1.67\times10^{-8}$ & $2.58\times10^{-8}$ & $\mathbf{3.42\times10^{-8}}$ & $-2.02\times10^{-8}$ & $2.08\times10^{-8}$ & $1.87\times10^{-8}$ & $1.12\times10^{-8}$ & $\mathbf{-4.91\times10^{-10}}$\\
\bottomrule
\end{tabular}%
}
\end{table*}

\begin{figure*}[t]
  \centering
  \includegraphics[width=1\linewidth]{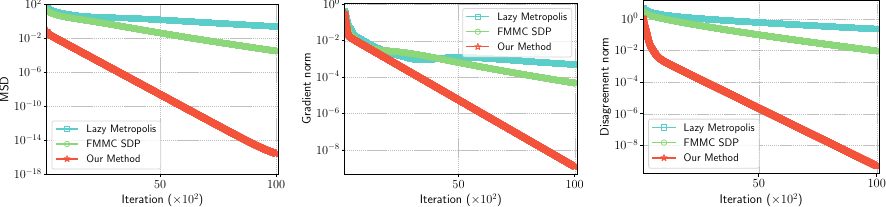}
  \caption{Time-varying graphs improve convergence without additional communication overhead. All curves are plotted starting at iteration $k=10^{2}$. For such a large and sparse graph as depicted in Fig.~\ref{fig:igs-graph}, the Lazy Metropolis and FMMC yield mixing weights with the second-largest eigenvalue $\lambda_2=0.989,0.979$, corresponding to approximate $\tau$-step contraction factors of $\epsilon_\tau \approx 0.86$ and $0.75$, respectively.}
  \label{fig:consensus-metropolis}
\end{figure*}

As shown in Figure~\ref{fig:gt-vs-centralized}, our method achieves centimeter-level parity with the centralized baseline in exemplified 10 epochs. The maximum Euclidean discrepancies are negligible: $6.88\times10^{-9}$\,m for positioning error and $5\times10^{-9}$\,m (light-travel distance) for satellite bias. This equivalence is further corroborated by the Mahalanobis distance \cite{mahalanobis2018generalized} $d_M(\mathbf{z})$ against ground truth (Table~\ref{tab1}), which utilizes the centralized precision matrix $\mathbf{\Sigma^{-1}}$ (see \textbf{Methods} on Remark \ref{def:r2} for the formal definition) to assess statistical consistency. As detailed in Table~\ref{tab1}, the divergence between the distributed and centralized solutions is restricted to the magnitude range $[4.91\times10^{-10}, 3.42\times10^{-8}]$, indicating extremely tight statistical agreement. Taken together, these results show that \textbf{\texttt{GraT-Diff}} effectively recovers the centralized solution—both in Euclidean and statistical terms—while relying solely on sparse, peer-to-peer communication.

\subsection*{GraT-Diff Accelerates Linear Convergence with Reduced Communication}
 To characterize both optimization and consensus behaviour, we monitor three per-iteration metrics: (i) the mean-square deviation (MSD) from the centralized solution, (ii) the average gradient norm, serving as a proxy for first-order optimality; and (iii) the network disagreement norm, quantifying the deviation of local nodes from the global average (see \textbf{Supplementary Information} for formal definitions). We compare \textbf{\texttt{GraT-Diff}} against two diffusion-based counterparts, both implemented within the same GT framework and using the communication pattern shown in Figure~\ref{fig:igs-graph}. The first uses heuristic Lazy Metropolis weights~\cite{nedic2017achieving}, while the second relies on weights obtained from the Fastest Mixing Markov Chain (FMMC) design~\cite{boyd2009fastest}, where a single mixing matrix is optimized via semidefinite programming (SDP) simultaneously promoting non-negative eigenvalues. The resulting fixed weight matrices have second-largest eigenvalues of $0.989$ and $0.979$, corresponding to approximate $\tau$-step contraction factors of $\epsilon_\tau \approx 0.86$ and $0.75$, respectively. The FMMC baseline can be viewed as optimizing the product in~\eqref{eq:prod} with a window length of $\tau=1$, i.e., as a single-snapshot design. Consequently, it requires every instantaneous communication graph to be connected. In contrast, the proposed approach only requires the $\tau$-block contraction ($\epsilon_\tau<1$), thereby allowing individual snapshots to be disconnected. Furthermore, our learned $\tau$-periodic schedule achieves $\epsilon_\tau = 0.68$ for $\tau=15$, which is expected to lead to faster consensus.
To illustrate, we report the convergence results for a representative epoch (epoch 7). As shown in Figure~\ref{fig:consensus-metropolis}, \textbf{\texttt{GraT-Diff}} exhibits a markedly steeper linear decay in all three metrics than either of the baselines, with theoretical justification provided in the \textbf{Supplementary Information}. In other words, to reach the same error tolerance relative to the centralized solution, our method requires far fewer iterations, which translates directly into a substantial reduction in communication overhead. The gain stems from the reduced contraction factor $\epsilon_\tau$ achieved by the learned, $\tau$-periodic schedule of mixing weights, which enables faster information diffusion over the same underlying network.

\subsection*{Structural Advantages over Standard Decentralized GNSS Network}
We provide a side-by-side comparison between our framework and the decentralized ADMM-based GNSS architecture of Khod.\&Teunissen~\cite{khodabandeh2019distributed}, summarized in Table~\ref{tab2}. The approach in Khod.\&Teunissen~\cite{khodabandeh2019distributed} typically relies on a hand-crafted, overlapping set of constraint nodes to enable intra-cluster hard averaging and promote convergence in practice. For large, time-varying networks such as in Figure~\ref{fig:igs-graph}, however, constructing and maintaining these overlapping clusters can be fragile, and may require manual retuning as the topology or receiver set evolves. By contrast, \textbf{\texttt{GraT-Diff}} uses a unified GT framework with learned mixing weights and does not require any explicit cluster construction or overlap design, making it inherently more robust to graph variations.

\begin{table}[t]
    \caption{Comparison with state-of-the-art decentralized GNSS framework}
    \label{tab2}
    \centering
    \begin{tabular} {l c c } 
        \toprule
        Framework & Ours & Khod.\&Teunissen \cite{khodabandeh2019distributed} \\
        \midrule
        Constraint set &  \pref{N}  &  \notpref{Y} \\
        Cluster averaging & \pref{N}    & \notpref{Y}    \\
        Time-var. graph & \pref{Y} & \notpref{N}  \\
        Convergence & \pref{Y} & \pref{Y} \\
        \bottomrule
    \end{tabular}
\end{table}

More broadly, our method belongs to the class of gradient-driven algorithms: each node updates using local gradients together with neighbour state exchanges, rather than the constraint-heavy primal–dual machinery of ADMM. This lightweight structure is more tolerant to packet losses, link failures, and bandwidth fluctuations, since missing messages perturb only incremental gradient steps instead of disrupting coupled subproblems. With a single round of peer-to-peer communication per iteration and rigorous convergence guarantees on sparse, time-varying graphs, the proposed framework is particularly advantageous for next-generation regimes—such as GNSS constellations augmented by massive LEO fleets. In these data-dense, bandwidth-constrained environments, the integration of learned, graph-aware mixing policies with a controllable contraction factor $\epsilon_\tau$ provides a principled and scalable pathway for decentralized global processing.

\section*{Discussion}
We have presented a fully decentralized architecture that resolves the GNSS network estimation problem using a \textbf{\texttt{GraT-Diff}} scheme. To our knowledge, this represents the first step toward realizing decentralized GNSS networks capable of operating over time-varying graphs. Our method attains the centralized solution through strictly local computations and neighbor exchanges, with convergence established at a linear rate. Furthermore, this work investigates the impact of graph structures on convergence speed, communication overhead, and system robustness—factors that are pivotal to addressing the big-data challenges inherent in modern satellite constellations.

The validation in this paper is simulation-based. Although the network geometry is constructed from actual IGS station coordinates and the satellite geometry is derived from RINEX navigation files, the carrier-phase and code observables are generated under a controlled measurement model. This setup isolates decentralized optimization error from GNSS modeling errors and allows direct comparison with the centralized BLUE solution; as such, the results do not yet represent a full real-data demonstration. Applying the proposed method to real IGS observation files requires a standard GNSS preprocessing pipeline, including cycle-slip detection and repair, outlier rejection, multipath mitigation, elevation- or carrier-to-noise-dependent weighting, antenna phase-center corrections, and refined atmospheric modeling. These factors affect the construction of the local matrices $\mathbf{A}_r$, $\mathbf{B}_r$, $\mathbf{Q}_r$, and $\mathbf{y}_r$, but do not alter the decentralized structure of the proposed algorithm. Furthermore, the present framework operates on the float ambiguity estimation problem and does not explicitly include integer ambiguity resolution. A real-data evaluation with robust preprocessing, distributed ambiguity management, and operational correction-product generation therefore remains an important direction for future work.

In addition, practical deployments must account for system-level impairments such as communication delays, asynchronous updates, bandwidth constraints, and message quantization. These factors may affect the effective information-mixing rate and introduce stale or incomplete updates, motivating extensions based on delay-robust gradient tracking, asynchronous diffusion, compressed communication, and adaptive graph scheduling. The present work focuses on batched network estimation within each epoch. Another important direction is to extend the framework to fully streaming operation, where gradient tracking and consensus are performed epoch-by-epoch as new GNSS measurements arrive. Such an extension would exploit the temporal structure of GNSS data, further reduce latency, and support continual decentralized learning. Since the proposed framework is inherently compatible with fast-switching topologies, it may also serve as a useful building block for future spaceborne, onboard, and hybrid terrestrial--space GNSS networks.

\section*{Methods}
\subsection*{Notation and Conventions}
Unless stated otherwise, we use the following notation throughout. Non-bold lowercase and uppercase letters denote scalars; bold lowercase letters denote column vectors; bold uppercase letters denote matrices. 
\(\mathbf{1}\) is the all-ones vector of appropriate size, and \(\mathbf{I}_d\) is the \(d\times d\) identity matrix. \(\mathrm{diag}(\cdot)\) returns a diagonal matrix, 
\(\mathrm{blkdiag}(\cdot)\) returns a block-diagonal matrix, and \(\mathrm{col}(\cdot)\) denotes vertical (column-wise) stacking of its arguments, respectively. 
We use \(\mathbb{E}(\cdot)\) and \(\mathbb{C}(\cdot,\cdot)\) to denote the expectation and covariance operators. 
The Hadamard and Kronecker products are denoted by \(\odot\) and \(\otimes\), respectively. 
We write \(\|\cdot\|\) for the Euclidean norm of a vector; for matrices, it denotes the spectral norm, $\|\cdot\|_{F}$ is the Frobenius norm, and $\|\cdot\|_{*}$ denotes the nuclear norm.

\subsection*{Observation Model and Decentralized Formulation}

The linearized carrier-phase and pseudo-range (code) GNSS observation equations are given by
\begin{equation}
\resizebox{.91\hsize}{!}{$
\begin{aligned}
\Delta \phi_{r, j}^s(i) = & \mathbf{g}_r^s(i)^{\mathrm{T}} \Delta \mathbf{p}_r(i)+\mathrm{d} t_r(i)+\lambda_j \delta_{r, j}(i) -\mathrm{d} t^s(i) -\lambda_j \delta_j^s(i) \\
&+ v_r^s(i)\tau_r(i) -\mu_j l_r^s(i)+\lambda_j a_{r, j}^s + \varepsilon_{\phi,r,j}^s(i), \\
\Delta \rho_{r, j}^s(i) = & \mathbf{g}_r^s(i)^{\mathrm{T}} \Delta \mathbf{p}_r(i)+\mathrm{d} t_r(i)+d_{r, j}(i)  -\mathrm{d} t^s(i) -d_j^s(i) \\
&+ v_r^s(i)\tau_r(i) +\mu_j l_r^s(i)+\varepsilon_{\rho,r,j}^s(i).
\end{aligned}
$}
\label{eq:gnss}
\end{equation}
Here,
\begin{itemize}
    \item $\Delta \phi_{r, j}^s(i)$ and $\Delta \rho_{r, j}^s(i)$ represent the observed-minus-computed undifferenced carrier-phase and pseudo-range observations, respectively.
    \item The indices $r$, $s$, $j$, and $i$ denote the receiver, satellite, frequency, and epoch indices, respectively.
    \item $\mathbf{g}_r^s(i)$ is the line-of-sight unit vector.
    \item $\Delta \mathbf{p}_r(i)$ denotes the incremental receiver position.
    \item $\mathrm{d} t_r(i)$ and $\mathrm{d} t^s(i)$ are the receiver and satellite clock offsets, respectively.
    \item $\delta_{r, j}(i)$ and $\delta_j^s(i)$ denote the receiver and satellite carrier-phase hardware biases, respectively.
    \item $d_{r, j}(i)$ and $d_j^s(i)$ denote the receiver and satellite pseudo-range hardware biases, respectively.
    \item $\tau_r(i)$ is the zenith tropospheric delay, with $v_r^s(i)$ being the corresponding tropospheric mapping function coefficient.
    \item $l_r^s(i)$ refers to the ionospheric delay on the first frequency. The coefficient $\mu_j = \lambda_j^2 / \lambda_1^2$ is defined using $\lambda_j$, the wavelength corresponding to frequency $j$.
    \item $a_{r, j}^s$ denotes the carrier-phase ambiguity.
    \item Finally, $\varepsilon_{\phi,r,j}^s(i)$ and $\varepsilon_{\rho,r,j}^s(i)$ denote the observation noise terms.
\end{itemize}

Consider a receiver network with nodes $r=1,\ldots,R$. Each node $r$ collects an observation vector $\mathbf{y}_r$ consisting of observed-minus-computed, undifferenced carrier-phase and pseudo-range measurements. The measurement noise $\mathbf{n}_r$ is assumed to have zero mean, $\mathbb{E}(\mathbf{n}_r)=\mathbf{0}$, and covariance $\mathbf{Q}_r$, with $\mathbb{C}(\mathbf{n}_r,\mathbf{n}_{r'})=\delta_{rr'}\,\mathbf{Q}_r$, where $\delta_{rr'}$ denotes the Kronecker delta.

In the present simulations, inter-system bias states between different GNSS constellations are not modeled explicitly. Instead, these effects are treated implicitly through the simplified receiver clock and hardware-bias terms adopted in the observation model. In addition, the stochastic observation model employs simplified weighting and does not explicitly incorporate constellation-dependent stochastic modeling across GPS, GLONASS, Galileo, and BeiDou observations. These assumptions allow the decentralized optimization and graph-based information-mixing behavior to be evaluated under controlled conditions while preserving the general multi-GNSS observation structure.

The estimation variables at receiver $r$ are organized into a local state $\mathbf{x}_r$ and a set of satellite-specific parameters collected in a global state $\mathbf{z}$. Different system models may be adopted. For instance, when some receivers serve as reference stations, their positions can be treated as known. Additional satellite-related parameters, such as orbit corrections, may also be included in the global state vector. Regardless of the specific modeling choices, the parameters can always be organized into a local state $\mathbf{x}_r$ and a global state $\mathbf{z}$. Under this representation, the linearized observation model admits the compact form \eqref{eq:obs_r} used in the \textbf{Results} section. It is important to note that $\mathbf{x}_r$ and $\mathbf{z}$ represent identifiable parameterizations derived from the original undifferenced observation model in \eqref{eq:gnss}. These parameterizations preserve the physical interpretability of the underlying parameters while ensuring full rank of the network design matrix. A detailed discussion is omitted here for brevity, we refer interested readers to \cite{teunissen2010ppp,hou2023decentralized,odijk2016estimability}.

A scalable decentralized formulation can be written as follows:
\begin{equation}
\begin{aligned}
& \underset{\{\mathbf{x}_r\}^{R}_{r=1},\ \mathbf{z}}{\text{minimize}} \quad f(\{\mathbf{x}_r\}^{R}_{r=1},\mathbf{z})=\sum_{r=1}^{R}f_{r}(\mathbf{x}_{r},\mathbf{z}),\\
 &\mathrm{where} \quad f_{r}(\mathbf{x}_{r},\mathbf{z}) = \frac{1}{2}\big\|\,\mathbf{A}_{r}\mathbf{x}_{r}+\mathbf{B}_{r}\mathbf{z}-\mathbf{y}_r\,\big\|_{\mathbf{Q}_{r}^{-1}}^2.
\end{aligned}
\label{eq:obs_dec}
\end{equation}
Under this formulation, each node minimizes its local objective $f_{r}(\mathbf{x}_{r},\mathbf{z})$ while exchanging only compact information with neighboring nodes.

\subsection*{Learning Graph-Aware Mixing Sequence}
The offline learning stage does not learn GNSS parameters from data. Instead, it learns how information should be mixed across the receiver graph. Mixing weight sequence $\theta$ can be learned by solving the surrogate problem
\begin{equation}
    \begin{aligned}
\underset{\theta=\{\mathbf{W}^{(l)}\}_{l=0}^{\tau-1}}{\text{minimize}}\quad
& \mathcal{L}=\|\mathrm{NN}(\cdot|{\theta})- \tfrac{1}{R}\,\mathbf{1}\mathbf{1}^{\top}\|_F^2 \\[2pt]
\text{subject to}\quad
& \mathrm{NN}(\cdot|{\theta})
= \mathbf{W}^{(\tau-1)}\cdots \mathbf{W}^{(1)}\mathbf{W}^{(0)},\\
& \mathbf{W}^{(l)}\mathbf{1}=\mathbf{1},\ \mathbf{1}^\top\mathbf{W}^{(l)}=\mathbf{1}^\top,\\
& [\mathbf{W}^{(l)}]_{rq}=0 \ \text{if } (r,q)\notin\mathcal{E}^{(l)},\\
& \mathbf{W}^{(l)}\ge 0,\qquad l=0,\dots,\tau-1.
\end{aligned}
\label{train}
\end{equation}
As we can see, gradients \(\{\nabla_{\mathbf{W}^{(l)}}\mathcal{L}\}_{l=0}^{\tau-1}\) could be computed by backpropagation
through the product. The main difficulty is enforcing, at every layer \(l\),
(i) nonnegativity, (ii) the graph-induced sparsity pattern, and (iii) double
stochasticity. To address this, we apply an exponentiated-gradient
(multiplicative) step—preserving nonnegativity—followed by a masked
Sinkhorn–Knopp projection \cite{knight2008sinkhorn,cuturi2013sinkhorn}—preserving sparsity and imposing double
stochasticity. Given the current loss \(\mathcal{L}\), the updates are
\begin{equation}
\begin{aligned}
 \widehat{\mathbf{W}}^{(l)} &=\mathbf{W}^{(l)}
\odot 
\exp\Big(-\eta\,\nabla_{\mathbf{W}^{(l)}}
\mathcal{L}\Big),\\
\widetilde{\mathbf{W}}^{(l)} &= \widehat{\mathbf{W}}^{(l)}\odot \mathbf{M}^{(l)}, \\
\mathbf{W}^{(l)} &=\mathbf{D}^{(l)}_{1}\,\widetilde{\mathbf{W}}^{(l)}\,\mathbf{D}^{(l)}_{2}.
\end{aligned}
\label{sinkhorn}
\end{equation}
Here, $\eta>0$ is the step size for the multiplicative update. $\mathbf{M}^{(l)}\in\{0,1\}^{R\times R}$ serves as a topological mask for $\mathcal{G}^{(l)}$, ensuring via Hadamard product that unconnected entries in $\widetilde{\mathbf{W}}^{(l)}$ are zeroed out while existing communication links and self-loops are preserved. Note that this masking framework offers flexibility for dynamic environments: $\mathbf{M}^{(l)}$ can either be updated explicitly to match changing communication protocols, or kept static while applying a magnitude threshold to the weights; the latter automatically prunes negligible connections to reduce communication overhead. Finally, $\mathbf{D}^{(l)}_{1}$ and $\mathbf{D}^{(l)}_{2}$ are diagonal matrices that scale the rows and columns to ensure the final mixing matrix $\mathbf{W}^{(l)}$ is doubly stochastic \footnote{The matrices $\mathbf{D}^{(l)}_{1}$ and $\mathbf{D}^{(l)}_{2}$ are not known a priori. $\mathbf{W}^{(l)} =\mathbf{D}^{(l)}_{1}\,\widetilde{\mathbf{W}}^{(l)}\,\mathbf{D}^{(l)}_{2}$ can be derived by alternating row and column normalizations on $\widetilde{\mathbf{W}}^{(l)}$: \(\widetilde{\mathbf{W}}^{(l)}\gets\mathrm{diag}(\widetilde{\mathbf{W}}^{(l)}\mathbf{1})^{-1}\widetilde{\mathbf{W}}^{(l)}\), \(\widetilde{\mathbf{W}}^{(l)}\gets\widetilde{\mathbf{W}}^{(l)}\mathrm{diag}(\mathbf{1}^\top\widetilde{\mathbf{W}}^{(l)})^{-1}\). After a few rounds (or upon convergence), set $\mathbf{W}^{(l)}\gets\widetilde{\mathbf{W}}^{(l)}$.}. Algorithm~\ref{alg:SINK} summarizes the procedure; see Fig.~\ref{fig_1} left for a visual recap.

To produce Figure~\ref{fig:training}, we choose the step size \(\eta=8\times10^3\), the maximum iteration \(\mathcal{K}=3\times10^{4}\), and uniform initialization $[0,1]$ on each $\mathbf{W}^{(l)}$, followed by the lazy transform $\frac{\mathbf{W}^{(l)}+\mathbf{I}}{2}$ as previously mentioned.

\begin{algorithm}[t]
\caption{Learning Mixing Sequence via Masked Sinkhorn--Knopp}
\label{alg:SINK}
\begin{algorithmic}[1]
\Require Window length \(\tau\); mask matrices \(\{\mathbf{M}^{(l)}\}_{l=0}^{\tau-1}\in\{0,1\}^{R\times R}\); stepsize \(\eta>0\); maximum training iterations \(\mathcal{K}\); initialization \(\{\mathbf{W}^{(l)}\ge 0\}\) with support \(\mathbf{M}^{(l)}\)
\Ensure Mixing matrices \(\theta=\{\mathbf{W}^{(l)}\}_{l=0}^{\tau-1}\)
\For{\(\kappa=0,1,\dots,\mathcal{K}-1\)}
  \State \(\mathrm{NN}(\cdot|{\theta})\gets \mathbf{W}^{(\tau-1)}\cdots \mathbf{W}^{(0)}\) \Comment{forward}
  \State \(\mathcal{L}\gets \|\mathrm{NN}(\cdot|{\theta})-\tfrac{1}{R}\,\mathbf{1}\mathbf{1}^{\top}\|_F^2\) \Comment{compute loss}
  \ForAll {\(l\in\{0,\dots,\tau-1\}\)} \Comment{in parallel}
    \State \(\widehat{\mathbf{W}}^{(l)} \gets \mathbf{W}^{(l)} \odot \exp\!\big(-\eta\,\nabla_{\mathbf{W}^{(l)}} \mathcal{L}\big)\) \Comment{backward}
    \State \(\widetilde{\mathbf{W}}^{(l)} \gets \widehat{\mathbf{W}}^{(l)} \odot \mathbf{M}^{(l)}\) \Comment{enforce sparsity}
    \State \(\mathbf{W}^{(l)} \gets\mathbf{D}^{(l)}_{1}\,\widetilde{\mathbf{W}}^{(l)}\,\mathbf{D}^{(l)}_{2}\) \\ \Comment{Sinkhorn--Knopp scaling}
  \EndFor
  \State \(\theta \gets \{\mathbf{W}^{(l)}\}_{l=0}^{\tau-1}\) \Comment{collecting}
\EndFor
\end{algorithmic}
\end{algorithm}

\begin{remark}[Implicit total support condition for the masked Sinkhorn--Knopp projection \cite{knight2008sinkhorn}] 
The Sinkhorn--Knopp projection onto the set of (masked) doubly stochastic matrices exists, and the alternating normalization converges, if and only if the input nonnegative matrix has \emph{total support}. In our setting, it suffices to check this property for the mask $\mathbf{M}^{(l)}$, since $\widetilde{\mathbf{W}}^{(l)}$ is strictly positive on the support of $\mathbf{M}^{(l)}$ and zero elsewhere via Hadamard product. In our settings (namely, $\mathbf{M}^{(l)}$ is symmetric and $[\mathbf{M}^{(l)}]_{rr}=1$ for all $r$), every positive entry of $\mathbf{M}^{(l)}$ lies on a positive diagonal: diagonal entries belong to the identity permutation, and off-diagonal pairs $(r,q)$ and $(q,r)$ belong to a permutation that swaps $r$ and $q$ while fixing all other indices. Hence $\mathbf{M}^{(l)}$ has total support, and so does $\widetilde{\mathbf{W}}^{(l)}$. It follows that the masked Sinkhorn--Knopp scaling exists, converges and preserves the sparsity pattern (acting blockwise if the mask decomposes from disconnected graphs).
\end{remark}

\begin{remark}[Spectral constraints for heterogeneity from partial visibility]\label{rm2} 
It is often desirable to constrain the eigenvalues of $\mathbf{W}^{(l)}$ to be nonnegative to mitigate large data heterogeneity and to allow for larger step sizes when using the learned $\theta$ in decentralized optimization, without collapsing the main diagonal entries (zeros will let the smallest eigenvalue hit the negative unity) \cite{nedic2017achieving,koloskova2021improved}. In GNSS networks, limited satellite visibility creates distinct local observation data, i.e., each node constructs its own matrices $\mathbf{A}_r$ and $\mathbf{B}_r$ for \eqref{eq:obs_dec} using only the satellites and signals it observes; consequently, a collapsed main diagonal (zero self-loop) is detrimental, as it implies a node completely discards its own unique construction in favor of neighbor averages. A simple way to solve this is to replace $\mathbf{W}^{(l)}$ with its lazy version $\frac{\mathbf{W}^{(l)} + \mathbf{I}}{2}$ in $\mathrm{NN}(\cdot|\theta)$, which can be interpreted as adding a residual connection \cite{he2016deep} to the deep linear network. A more nuanced approach is to augment the loss $\mathcal{L}$ with the regularizer $\Omega(\mathbf{W}^{(l)}) = \|\mathbf{W}^{(l)}\|_{*} - \operatorname{tr}(\mathbf{W}^{(l)})$, which simplifies to $\sum_{r=1}^{R} (|\lambda_r| - \lambda_r)$ to explicitly penalize negative eigenvalues. Crucially, this term is amenable to standard backpropagation: the trace is linear, and the nuclear norm $\|\mathbf{W}^{(l)}\|_{*}$ is differentiable almost everywhere with gradient $\frac{\partial \|\mathbf{W}^{(l)}\|_{*}}{\partial \mathbf{W}^{(l)}} = \mathbf{W}^{(l)} \big((\mathbf{W}^{(l)})^{\top}\mathbf{W}^{(l)}\big)^{-\frac{1}{2}}$ \cite{watson1992characterization} under standard regularity conditions.
\end{remark}

\subsection*{Diffusion Algorithm}
Equipped with the learned weight sequence $\theta$, the \textbf{\texttt{GraT-Diff}} algorithm solves \eqref{eq:obs_dec} iteratively. At iteration $k$, each receiver node $r$ maintains a state variable $\mathbf{z}_{r}^{k}$ and a gradient tracker $\mathbf{g}_{r}^{k}$. The update cycle proceeds in three stages: consensus-based descent, local variable recovery, and gradient tracking adaptation.

First, the receivers update $\mathbf{z}_{r}^{k+1}$ by combining neighbor information with the descent direction provided by the tracker. Second, the local auxiliary variable $\mathbf{x}_{r}^{k+1}$ is recovered via the closed-form solution to the local subproblem $\min_{\mathbf{x}_{r}} f_{r}(\mathbf{x}_{r}|\mathbf{z}_{r}^{k+1})$. Finally, the tracker $\mathbf{g}_{r}^{k+1}$ is updated to estimate the global gradient mean. The full update laws are:
\begin{equation}
\begin{aligned}
     \mathbf{z}_{r}^{k+1}&= \sum _{q\in\mathcal{N}^{(l)}_{r}} w^{(l)}_{qr}(\mathbf{z}_{q}^{k}-\mu\mathbf{g}^{k}_{q}),\\
     \mathbf{x}_{r}^{k+1}& = (\mathbf{A}^\top_{r}\mathbf{Q}_{r}^{-1}\mathbf{A}_{r})^{-1}\mathbf{A}^\top_{r}\mathbf{Q}_{r}^{-1}(\mathbf{y}_{r}-\mathbf{B}_{r}\mathbf{z}_{r}^{k+1}), \\
      \mathbf{g}_{r}^{k+1}& =\sum _{q\in\mathcal{N}^{(l)}_{r}}w^{(l)}_{qr}\mathbf{g}^{k}_{q} + \nabla_{\mathbf{z}} f_{r}(\mathbf{x}_{r}^{k+1},\mathbf{z}_{r}^{k+1})-\nabla_{\mathbf{z}} f_{r}(\mathbf{x}_{r}^{k},\mathbf{z}_{r}^{k}),
\end{aligned}
 \label{eq:diffusionGT}
\end{equation}
where $l=k\%\tau$ (i.e., $k\bmod\tau$), $\mu>0$ is the stepsize, and the tracker is initialized as $\mathbf{g}^{0}_{r}=\nabla_{\mathbf{z}} f_{r}(\mathbf{x}_{r}^{0},\mathbf{z}_{r}^{0})$. 

Crucially, the middle line in~\eqref{eq:diffusionGT} represents the explicit solution to $\mathbf{x}_{r}^{k+1}=\arg\min_{\mathbf{x}_r} f_r(\mathbf{x}_r |\mathbf{z}_r^{k+1})$. It is worth noting that while the original local loss in \eqref{eq:obs_dec} may be ill-posed with respect to $\mathbf{x}_r$ alone, conditioning on the consensus variable $\mathbf{z}_r^{k+1}$ acts as a regularizer, rendering the subproblem strictly convex and uniquely solvable. By cyclically applying the weights from $\theta$, the algorithm leverages the block contraction property established in \eqref{eq:prod}, thereby accelerating both state consensus and the convergence of the gradient tracker to the true global gradient. A high-level implementation is given in Algorithm~\ref{alg:gt-atc}; also see Fig.~\ref{fig_1} right for a visual recap.

It is worth noting that the diffusion form in~\eqref{eq:diffusionGT} is inherently robust to unintended communication-link interruptions, during which the scheduled mixing matrices may no longer be applied exactly. In such situations, each node dynamically constructs a fallback mixing matrix $[w^{(l)}_{rq}]$ using the currently active communication graph, i.e., the set of neighbors from which messages are successfully received at the current iteration, which is equivalent to this fallback matrix following a standard heuristic Metropolis construction as shown in~\cite{nedic2017achieving}. Therefore, a temporarily unavailable link simply receives zero weight during the affected iterations, while the diagonal entries are adjusted to preserve stochasticity. No explicit reinitialization or retransmission mechanism is required. The propagated node states continue evolving from their current values using the active communication graph, such that historical information is retained implicitly in the local states, whereas missed messages are not buffered or replayed after reconnection. Once the link is restored, it automatically resumes the learned mixing schedule $\theta$. Such link interruptions are commonly studied in highly unreliable communication environments, where the resulting time-varying topology is used to model random communication failures rather than to optimize information exchange. Under the $B$-connected graph assumption, convergence can still be guaranteed even in this worst-case setting, although such analysis does not provide improvements in convergence rate or communication efficiency, as established in~\cite{nedic2017achieving}.

Finally, to produce Figure~\ref{fig:gt-vs-centralized} and Figure~\ref{fig:consensus-metropolis}, we choose a stepsize \(\mu=0.015\), initialize each node's state $\mathbf{z}_r^0=\mathbf{0}$ and limit the number of iterations to \(K=10^{4}\).

\begin{algorithm}[t]
\caption{\textbf{\texttt{GraT-Diff}} for Decentralized GNSS}
\label{alg:gt-atc}
\begin{algorithmic}[1]
\Require Mixing matrices $\theta=\{\mathbf{W}^{(l)}\}_{l=0}^{\tau-1}$ satisfying \eqref{eq:prod} with entries $w^{(l)}_{qr}$; stepsize $\mu>0$; maximum iterations $K$
\Require Local data $\{\mathbf{A}_r,\mathbf{B}_r,\mathbf{Q}_{r},\mathbf{y}_r\}_{r=1}^R$ and $f_r$ as in \eqref{eq:obs_dec}
\Ensure Estimates $\{\mathbf{x}_r^{K},\mathbf{z}_r^{K}\}_{r=1}^R$
\State \textbf{Initialization:} For each node $r$, choose $\mathbf{z}_r^0$ and set $\mathbf{x}_r^0 \gets (\mathbf{A}_r^\top \mathbf{Q}_{r}^{-1}\mathbf{A}_r)^{-1}\mathbf{A}_r^\top \mathbf{Q}_{r}^{-1}(\mathbf{y}_r-\mathbf{B}_r\mathbf{z}_r^0)$, $\mathbf{g}_r^0 \gets \nabla_{\mathbf{z}} f_r(\mathbf{x}_r^0,\mathbf{z}_r^0)$ 
\For{$k=0$ \textbf{to} $K-1$}
  \State $l \gets k \% \tau$ \Comment{select current topology $\mathbf{W}^{(l)}$}
  \ForAll{$r \in \{1,\ldots,R\}$} \Comment{in parallel}
    \State $\mathbf{z}_r^{k+1} \gets \sum_{q\in\mathcal{N}^{(l)}_r} w^{(l)}_{qr}(\mathbf{z}_q^k - \mu\,\mathbf{g}_q^k)$  \\ \Comment{adapt + combine state}
    \State $\mathbf{x}_r^{k+1} \gets (\mathbf{A}_r^\top \mathbf{Q}_{r}^{-1}\mathbf{A}_r)^{-1}\mathbf{A}_r^\top \mathbf{Q}_{r}^{-1}(\mathbf{y}_r-\mathbf{B}_r\mathbf{z}_r^{k+1})$  \\ \Comment{local update by exact minimizer}
    \State $\mathbf{g}_r^{k+1} \gets \sum_{q\in\mathcal{N}^{(l)}_r} w^{(l)}_{qr}\,\mathbf{g}_q^k + \nabla_{\mathbf{z}} f_{r}(\mathbf{x}_{r}^{k+1},\mathbf{z}_{r}^{k+1})-\nabla_{\mathbf{z}} f_{r}(\mathbf{x}_{r}^{k},\mathbf{z}_{r}^{k})$  \Comment{combine tracker + gradient tracking}
  \EndFor
\EndFor
\end{algorithmic}
\end{algorithm}

\begin{remark}[Analytical reduction]
   As each update sets $\mathbf x_r^{k}$ to its exact minimizer, it is convenient for the \emph{analysis} to eliminate $\mathbf x_r$ and rewrite the local objective $f_{r}(\mathbf{x}_{r},\mathbf{z})$ solely as a function of $\mathbf z$:
\begin{equation}
\begin{aligned}
f_{r}(\mathbf{z})&= \frac{1}{2}\big\|\,\mathbf{C}_{r}\big(\mathbf{B}_{r}\mathbf{z}-\mathbf{y}_r\big)\big\|_{\mathbf{Q}_{r}^{-1}}^2, \\
\mathbf{C}_{r}&\triangleq \mathbf{I}-\mathbf{A}_{r}(\mathbf{A}^\top_{r}\mathbf{Q}_{r}^{-1}\mathbf{A}_{r})^{-1}\mathbf{A}^\top_{r}\mathbf{Q}_{r}^{-1}.
\end{aligned}
\label{eq:close2}
\end{equation}
Its gradient w.r.t $\mathbf z$ is therefore
\begin{equation}
\nabla f_r(\mathbf z)
= \mathbf B_r^\top \mathbf C_r^\top \mathbf Q_r^{-1}\mathbf C_r\big(\mathbf B_r \mathbf z - \mathbf y_r\big).
\label{eq:local-grad}
\end{equation}
Note that in the implementation, $\mathbf x_r^{k+1}$ is still computed each round as in~\eqref{eq:diffusionGT} to evaluate the local gradient and to enable real-time monitoring or enforcement of application-specific quantities that depend on $\mathbf x_r$. Furthermore, the precision matrix can be derived from \eqref{eq:local-grad} as the aggregation $\mathbf{\Sigma^{-1}}=\sum_r \mathbf{B}_r^{\top} \mathbf{C}_r^{\top} \mathbf{Q}_r^{-1} \mathbf{C}_r \mathbf{B}_r$, which corresponds to the centralized BLUE information matrix.
 \label{def:r2}
\end{remark}
Note that Remark \ref{def:r2} provides an elimination equivalence for the objective function, which serves as the foundation for our theoretical analysis. For a rigorous investigation into the linear convergence of the proposed method, we refer interested readers to the \textbf{Supplementary Information}.

\section*{Acknowledgements}
The authors would like to thank Prof. Peter J. G. Teunissen for his insightful discussions during the development of the distributed GNSS framework. His conceptual suggestions greatly contributed to the early stages of this work. This work was supported by the King Abdullah University of Science and Technology (KAUST) Office of Sponsored Research (OSR) under Award No. RFS-CRG12-2024-6478. 

\section*{Author Contributions}
Xue Xian Zheng carried out the algorithmic development and implementation, performed the theoretical analysis, ran simulations, generated figures, and co-drafted the initial manuscript.
Xing Liu led the modelling/problem formulation, implemented the methods and coding, coordinated the project, co-drafted the initial manuscript, and served as the corresponding author.
Tareq Y. Al-Naffouri supervised the project and revised the manuscript. 
All authors discussed the results and approved the final manuscript.

\section*{Competing Interests}
Tareq Y. Al-Naffouri serves as a Guest Editor for the \textit{npj Wireless Technology} Collection “Advances in Wireless Positioning and Sensing.” He had no role in the peer review, editorial decision-making, or handling of this manuscript, which was managed independently by another editor with no competing interests. The remaining authors declare no competing interests.

\section*{Data Availability}The simulated datasets and the accompanying scenario generator will be deposited in a public repository upon publication (DOI to be provided). The input data consist of IGS station coordinates and multi-GNSS broadcast navigation files, which are publicly available from the respective providers and mirror archives.

\section*{Code Availability}A reference implementation of \textbf{\texttt{GraT-Diff}}, including Algorithm 1\&2 to reproduce the figures, will be deposited in a public repository upon publication (DOI to be provided) under an open source license. Configuration files and environment/requirements files will be included to facilitate reproduction.

{
\bibliographystyle{IEEEtran}
\bibliography{reg}
}

%

\clearpage

\twocolumn[{%
\vspace*{3em}

\centering
{\normalfont\sffamily\Huge
Supplementary Information: Proof of Linear\\[0.25em]
Convergence for \textbf{\texttt{GraT-Diff}}\par}

\vspace{5em}

\rule{0.22\textwidth}{0.4pt}
\hspace{1em}
{\large$\blacklozenge$}
\hspace{1em}
\rule{0.22\textwidth}{0.4pt}

\vspace{4em}
}]

\vspace{1em}

\setcounter{equation}{0}
\renewcommand{\theequation}{S\arabic{equation}}
\renewcommand{\appendixname}{Note}
\renewcommand{\theremark}{S1}
\renewcommand{\thefigure}{S1}

\IEEEdisplaynontitleabstractindextext

%


In this note, we present a gentle and self-contained convergence analysis of the proposed \textbf{\texttt{GraT-Diff}} algorithm, tailored to the GNSS setting. The development proceeds in three steps: (i) we state the convergence target and assumptions; (ii) we derive the key error recursions and supporting lemmas; and (iii) we establish the main linear-convergence result and discuss the associated trade-offs and parameter choices.

Before proceeding with these steps, it is helpful to establish a concise notation. Recalling Remark 3, let each local vector be denoted by $\mathbf z_r\in\mathbb R^d$. This allows us to express the analytical network recursion of (7) compactly as:
\begin{equation}
\begin{aligned}
\mathcal{Z}^{k+1} &= \mathcal{W}^k\big(\mathcal{Z}^{k}-\mu\,\mathit{G}^{k}\big),\\
\mathit{G}^{k+1} &= \mathcal{W}^k\,\mathit{G}^{k} + \nabla \mathcal{F}^{\,k+1}-\nabla \mathcal{F}^{\,k},
\end{aligned}
\label{eq:gt_vec}
\end{equation}
where the augmented quantities are
\begin{equation}
\begin{aligned}
\mathcal{Z}^{k} & \triangleq \mathrm{col}\big(\mathbf{z}_1^{k},\ldots,\mathbf{z}_R^{k}\big)\in\mathbb{R}^{Rd},\\
\mathit{G}^{k} & \triangleq \mathrm{col}\big(\mathbf{g}_1^{k},\ldots,\mathbf{g}_R^{k}\big)\in\mathbb{R}^{Rd},\\
\mathcal{W}^k &\triangleq \mathbf{W}^{(k\%\tau)}\otimes \mathbf{I}_d\in\mathbb{R}^{Rd\times Rd},\\
\nabla \mathcal{F}^{\,k} &\triangleq \mathrm{col}\big(\nabla f_1(\mathbf{z}_1^{k}),\ldots,\nabla f_R(\mathbf{z}_R^{k})\big)\in\mathbb{R}^{Rd}.
\end{aligned}
\label{eq:together}
\end{equation}
In addition, define the network averages (centroids) of $\mathcal Z^{k}$, $\mathit G^{k}$, and $\nabla \mathcal F^{\,k}$ as
\begin{equation}
\begin{aligned}
\bar{\mathbf{z}}^{k} &\triangleq \tfrac{1}{R}\sum \mathbf{z}_r^{k} = \tfrac{1}{R}\big(\mathbf{1}^{\top}\!\otimes\!\mathbf{I}_d\big)\mathcal{Z}^{k}\in\mathbb{R}^{d},\\
\bar{\mathbf{g}}^{k} &\triangleq \tfrac{1}{R}\sum \mathbf{g}_r^{k} = \tfrac{1}{R}\big(\mathbf{1}^{\top}\!\otimes\!\mathbf{I}_d\big)\mathit{G}^{k}\in\mathbb{R}^{d},\\
\overline{\nabla f}^{\,k} &\triangleq \tfrac{1}{R}\sum \nabla f_r(\mathbf{z}_r^{k}) =  \tfrac{1}{R}\big(\mathbf{1}^{\top}\!\otimes\!\mathbf{I}_d\big)\nabla \mathcal{F}^{\,k}\in\mathbb{R}^{d}.
\end{aligned}
\label{eq:avere_vec}
\end{equation}

\section*{Necessary Assumptions}
The analysis of \eqref{eq:gt_vec} aims to show that $\mathcal{Z}^{k} \to \mathbf{1} \otimes \mathbf{z}^\star$ as $k \to \infty$ (under suitable norms, for example \textbf{MSD} $\Vert \mathcal{Z}^{k} - \mathbf{1} \otimes \mathbf{z}^{\star} \Vert^{2}/R$), where $\mathbf{z}^\star$ denotes the centralized minimizer. We base our analysis on a standard GT framework within the smooth and strongly convex regime, which allows us to characterize convergence rates in terms of condition numbers. The standing assumptions are formalized below.

\begin{assumption}[Strong convexity and smoothness] 
\label{as:smooth-strong}Each local loss function $f_r(\mathbf{z})$ from (8) is $L$-smooth with $L=\sup_r\lambda_{\max}(\mathbf{H}_{r})$ and Hessian $\mathbf{H}_r = \nabla^2 f_r(\mathbf{z}) \succeq \mathbf{0}$.
The global objective $F(\mathbf{z})=\tfrac{1}{R}\sum_{r=1}^R f_r(\mathbf{z})$ is $m$-strongly convex with $m = \lambda_{\min}(\mathbf{H})$ and aggregate Hessian
$\mathbf{H}=\frac{1}{R}\sum_{r=1}^R \mathbf{H}_r \succ \mathbf{0}$. $Q=L/m\ge1$ is the condition number.
\end{assumption}

\begin{assumption}[Block contraction over $\tau$ steps]
\label{as:stochastic} For all $l$, $\mathbf{W}^{(l)}$ is doubly stochastic (i.e., $\,\mathbf{W}^{(l)}\mathbf{1}=\mathbf{1},\ \mathbf{1}^{\mathsf T}\mathbf{W}^{(l)} \!=\! \mathbf{1}^{\mathsf T}$)
and respects the sparsity pattern of $\mathcal{G}^{(l)}$. There exists $\epsilon_\tau\!\in\!(0,1)$ such that
$\big\|\mathbf{W}^{(\tau-1)}\!\cdots\!\mathbf{W}^{(0)} - \tfrac{1}{R}\,\mathbf{1}\mathbf{1}^{\mathsf T}\big\| = \epsilon_\tau$ as in (2).
\end{assumption}

\section*{Error Recursion and Lemmas}
Rather than working directly with the limit $\mathcal{Z}^{k}\!\to\!\mathbf{1}\!\otimes\!\mathbf{z}^\star$, we introduce the consensus projector
$\mathcal{J}=\tfrac{1}{R}\,(\mathbf{1}\mathbf{1}^{\mathsf T}) \otimes \mathbf{I}_d \in \mathbb{R}^{Rd\times Rd}$ and its orthogonal complement
$\mathcal{J}_\perp=\mathbf{I}_{Rd}-\mathcal{J}\in\mathbb{R}^{Rd\times Rd}$, and define the error quantities:
\begin{equation}
\begin{aligned}
    & \mathcal{C}^k = \mathcal{J}_\perp \mathcal{Z}^k,\qquad
\mathcal{T}^k=\mathcal{J}_\perp \mathit{G}^k, \\
&\mathbf{s}^{k}= \bar{\mathbf{z}}^{k}-\mathbf{z}^\star 
\quad\text{and}\quad
\Delta^k = \bar{\mathbf{g}}^k - \nabla F(\bar{\mathbf{z}}^k).
\end{aligned}
\label{notations}
\end{equation}
Here, $\mathcal{C}^k$ (\textbf{network disagreement}) measures deviation from the consensus subspace; $\mathcal{C}^k=\mathbf{0}$ iff all agents share the same iterate.
$\mathcal{T}^k$ (tracking disagreement) measures disagreement among gradient trackers; $\mathcal{T}^k=\mathbf{0}$ iff the trackers are consensual.
$\mathbf{s}^k$ (average optimality error) vanishes iff the network average equals the optimizer $\mathbf{z}^\star$.
$\Delta^k$ (inexact gradient error) is the bias of $\bar{\mathbf{g}}^k$ (\textbf{averaged gradient}) at $\bar{\mathbf{z}}^k$; $\Delta^k=\mathbf{0}$ under exact tracking or noise-free gradients.

With these notations, convergence to a consensual optimizer is equivalent to the joint vanishing of the errors. Under suitable norms, we have
\begin{equation}
\mathcal{C}^k \to \mathbf{0},\quad \mathcal{T}^k \to \mathbf{0},\quad \mathbf{s}^k \to \mathbf{0}\quad \mathrm{as} \quad (k\to\infty),
\end{equation}
in the presence of gradient error $\Delta^k$. Based on this understanding, we derive the following lemmas.

\begin{lemma}[Average identities]
\label{lem:avg-invariants}
Under Assumption~\ref{as:smooth-strong}--\ref{as:stochastic}, the following hold for all $k\!\ge\!0$:
\begin{align}
\bar{\mathbf{z}}^{k+1} &= \bar{\mathbf{z}}^{k} - \mu\,\bar{\mathbf{g}}^{k}, \label{eq:avg-z}\\
\bar{\mathbf{g}}^{k+1} &= \bar{\mathbf{g}}^{k} + \overline{\nabla f}^{k+1} -\overline{\nabla f}^k = \overline{\nabla f}^{k+1}. \label{eq:avg-g}
\end{align}
Moreover, it follows that
\begin{equation}\label{eq:delta-bound}
\begin{aligned}
    \|\Delta^k\|
&= \left\|\overline{\nabla f}^{k}- \nabla F(\bar{\mathbf{z}}^k)\right\| \\ &= \left\|\frac1R\sum_{r=1}^R \big(\nabla f_r(\mathbf{z}_r^{k}) - \nabla f_r(\bar{\mathbf{z}}^{k})\big)\right\|
\le \frac{L}{\sqrt{R}}\,\|\mathcal{C}^k\|.
\end{aligned}
\end{equation}
\end{lemma}
\begin{proof}
See Note S\ref{Note:A}.
\end{proof}

\begin{lemma}[Average descent]\label{lem:avg-descent}
Under Assumption~\ref{as:smooth-strong}, for any $\mu\in(0,1/L]$, we have
\begin{equation}\label{eq:avg-contract}
\begin{aligned}
    \|\mathbf{s}^{k+1}\|
&\le (1- m\mu)\,\|\mathbf{s}^{k}\| + \mu\,\|\Delta^k\| \le \alpha_{1}\,\|\mathbf{s}^{k}\| + \alpha_{2}\|\mathcal{C}^k\|,
\end{aligned}
\end{equation}
where 
$\alpha_{1} =1- m \mu$ and $\alpha_{2} = \mu \frac{L}{\sqrt{R}}.$

\end{lemma}
\begin{proof}
See Note S\ref{Note:B}.
\end{proof}

Lemma~\ref{lem:avg-invariants} states that in our \textbf{\texttt{GraT-Diff}}, the centroids
$(\bar{\mathbf{z}}^k,\bar{\mathbf{g}}^k)$ evolve like a centralized update on $F(\mathbf{z})$ but with a perturbation. The only difference is that the global gradient at the average, $\nabla F(\bar{\mathbf{z}}^k)$, is replaced by the network-average gradient $\overline{\nabla f}^{k}$. The resulting bias is precisely $\Delta^k$, and \eqref{eq:delta-bound} shows that its norm is upper bounded by the scaled disagreement norm $\|\mathcal{C}^k\|$. 
Building on this, Lemma~\ref{lem:avg-descent} establishes that, under the centroid dynamics of Lemma~\ref{lem:avg-invariants}, the average optimality error $\|\mathbf{s}^k\|$ contracts by a factor $\alpha_1$ up to an additive term $\alpha_2\|\mathcal{C}^k\|$. Consequently, a reduction in $\|\mathcal{C}^k\|$ is essential to drive the decay of $\|\mathbf{s}^k\|$. We now analyze the dynamics of the disagreement $\mathcal{C}^k$.

\begin{lemma}[Network disagreement recursion]\label{lem:block-consensus}
For recursion over window period $\{(i-1)\tau,\cdots,k-1,k\}$ where $k\in \mathcal{I}_{i}=\{i\tau,\ldots,(i+1)\tau-1\}$, the network disagreement $\mathcal{C}^k$ satisfies
\begin{equation}\label{eq:block-e}
\|\mathcal{C}^k\|
\le \epsilon_\tau \|\mathcal{C}^{(i-1)\tau}\| + \mu\, \sum_{j=(i-1)\tau}^{k-1}\| \mathcal{T}^j \|.
\end{equation}
\end{lemma}

\begin{proof}
See Note S\ref{Note:C}.
\end{proof}

\begin{lemma}[Tracking error recursion]\label{lem:track-rec}
For recursion over window period $\{(i-1)\tau,\cdots,k-1,k\}$, where $k\in \mathcal{I}_{i}=\{i\tau,\ldots,(i+1)\tau-1\}$, the tracking disagreement $\mathcal{T}^k $ satisfies
\begin{equation}\label{eq:track-one}
\begin{aligned}
 \Vert\mathcal{T}^{k}\Vert 
&\le
\epsilon_\tau \Vert\mathcal{T}^{(i-1)\tau}\Vert
+ \beta_{1}\sum_{j=(i-1)\tau}^{k-1}
\Vert\mathcal{C}^{\,j}\Vert \ + \beta_{2}\sum_{j=(i-1)\tau}^{k-1}\Vert\mathcal{T}^{\,j}\Vert
\\ & + \beta_{3}\sum_{j=(i-1)\tau}^{k-1}\Vert\mathbf{s}^{\,j}\Vert,
\end{aligned}
\end{equation}
where 
\begin{equation}\label{eq:track-two}
\begin{aligned}
\beta_{1} &=L\,c_W+\mu L^{2},\quad \beta_{2} =\mu\,L,\\
\beta_{3} &=\mu\,\sqrt{R}\,L^{2}, \quad c_W = \sup_j \big\|(\mathbf{W}^{(j \%\tau)}-\mathbf{I})\big\|_{\mathcal{J}_\perp}\leq 2.
\end{aligned}
\end{equation}
\end{lemma}

\begin{proof}
See Note S\ref{Note:D}.
\end{proof}

Instead of analyzing single iterations, Lemmas~\ref{lem:block-consensus} and \ref{lem:track-rec} study a block iterate that spans one full period of the mixing combination matrices, as specified in Assumption~\ref{as:stochastic}, where per-iteration graph connectivity may fail. These lemmas show that, although Assumption~\ref{as:stochastic} guarantees a block-wise contraction between the beginning and end of each period, it also introduces summations across the intermediate steps within the block, complicating the analysis. Consequently, the standard small-increment argument is impractical, and the sections that follow develop an alternative analysis framework tailored to this setting.

\section*{Convergence Results}
We define $T_{i} = \max_{k\in\mathcal{I}_{i}}\Vert\mathcal{T}^{k}\Vert,\, C_{i} = \max_{k\in\mathcal{I}_{i}}\Vert\mathcal{C}^{k}\Vert,\, S_{i} = \max_{k\in\mathcal{I}_{i}}\Vert\mathbf{s}^{k}\Vert$. Lemma \ref{lem:track-rec} can therefore be expressed as
\begin{equation}\label{eq:block-T-rev}
\begin{aligned}
T_i  &\le\epsilon_\tau T_{i-1}+\tau \beta_{1}C_{i-1} +(\tau-1) \beta_{1}C_{i} + \tau \beta_{2} T_{i-1} \\& + (\tau-1) \beta_{2}T_{i} + \tau \beta_{3} S_{i-1} + (\tau-1) \beta_{3}S_{i},
\end{aligned}
\end{equation}
with $\max_{k\in\mathcal{I}_{i}}\Vert\mathcal{T}^{(i-1)\tau}\Vert \leq T_{i-1}$, and 
\begin{equation}\label{eq:block-T-rev2}
\begin{aligned}
 &\max_{k\in\mathcal{I}_{i}}\sum_{j=(i-1)\tau}^{k-1}\Vert\mathcal{C}^{j}\Vert \\&\leq \max_{k\in\mathcal{I}_{i}}\sum_{j=(i-1)\tau}^{i\tau-1}\Vert\mathcal{C}^{j}\Vert+\max_{k\in\mathcal{I}_{i}}\sum_{j=i\tau}^{k-1}\Vert\mathcal{C}^{j}\Vert \\
 & \leq \sum_{j=(i-1)\tau}^{i\tau-1}C_{i-1}+\max_{k\in\mathcal{I}_{i}}\sum_{j=i\tau}^{(i+1)\tau-2}\Vert\mathcal{C}^{j}\Vert \leq \tau C_{i-1} +(\tau-1)C_{i}
\end{aligned}
\end{equation}
(same for $\max_{k\in\mathcal{I}_{i}}\sum_{j=(i-1)\tau}^{k-1}\{\Vert\mathcal{T}^{j}\Vert,\, \Vert\mathbf{s}^{j}\Vert\}$). Lemma \ref{lem:block-consensus} is
\begin{equation}\label{eq:block-C}
C_i \le \epsilon_\tau C_{i-1} + \tau\mu T_{i-1}+(\tau-1)\mu T_{i} ,
\end{equation}
and Lemma \ref{lem:avg-descent} can be rewritten as
\begin{equation}\label{eq:block-S}
\begin{aligned}
    S_{i} & \le\alpha_{1}^{\tau} S_{i-1} + \tau \alpha_2 C_{i-1}+(\tau-1) \alpha_2 C_{i}\\
& \le (1-\frac{m\tau}{2}\mu) S_{i-1} + \tau \alpha_2 C_{i-1}+(\tau-1) \alpha_2 C_{i}
\end{aligned}
\end{equation}
for $\mu\le \frac{1}{m\tau}$.
Moreover, with $\mathbf{v}_i=[c_W \tau C_i,\,L^{-1}T_i,\,\sqrt{R}S_{i}]^\top$, we can equivalently express \eqref{eq:block-T-rev}, \eqref{eq:block-C} and \eqref{eq:block-S} as
    \begin{equation}\label{eq:block-all}
\begin{aligned}
   \mathbf{v}_i &\le\Bigg( \underbrace{ \begin{bmatrix}
\epsilon_{\tau} &
0&
0\\
1 & \epsilon_{\tau} & 0\\
0 &
0 &
1
\end{bmatrix}}_{\boldsymbol{\varphi_{0}}}+\underbrace{ \begin{bmatrix}
0&
c_W\tau&
0\\
\frac{1}{c_W\tau} & 1 & 1\\
\frac{1}{c_W\tau} &
0 &
\frac{-1}{2Q}
\end{bmatrix}}_{\boldsymbol{\varphi_{1}}}\tau L\mu\Bigg)\mathbf{v}_{i-1} \\ &+\Bigg( \underbrace{ \begin{bmatrix}
0 &
0&
0\\
1-\frac{1}{\tau} & 0 & 0\\
0 &
0 &
0
\end{bmatrix}}_{\boldsymbol{\psi_{0}}}+\underbrace{ \begin{bmatrix}
0&
c_W\tau&
0\\
\frac{1}{c_W\tau}  & 1 & 1\\
\frac{1}{c_W\tau}  &
0 &
0
\end{bmatrix}}_{\boldsymbol{\psi_{1}}}(\tau-1) L\mu\Bigg)\mathbf{v}_{i}.
\end{aligned}
\end{equation}
Let $\omega=L\mu<1$, \eqref{eq:block-all} can be further written as
\begin{equation}\label{eq:block-iir}
\begin{aligned}
\mathbf{v}_i\le\big(\mathbf{I}-\boldsymbol{\psi_{0}}-\boldsymbol{\psi_{1}}(\tau-1)\omega\big)^{-1}\big(\boldsymbol{\varphi_{0}}+\boldsymbol{\varphi_{1}}\tau\omega\big)\mathbf{v}_{i-1},
\end{aligned}
\end{equation}
provided that $\mathbf{I}-\boldsymbol{\psi_{0}}-\boldsymbol{\psi_{1}}(\tau-1)\omega$ is invertible. To streamline notations, we introduce
\begin{equation}\label{eq:block-continuous}
\begin{aligned}
& \boldsymbol{\mathcal{H}}(\omega) =\boldsymbol{\mathcal{R}}(\omega)\boldsymbol{\varphi}(\omega), \quad \boldsymbol{\mathcal{R}}(\omega)=\big(\mathbf{I}-\boldsymbol{\psi}(\omega)\big)^{-1},\\
& \boldsymbol{\psi}(\omega) = \boldsymbol{\psi_{0}}+\boldsymbol{\psi_{1}}(\tau-1)\omega, \quad  \boldsymbol{\varphi}(\omega) = \boldsymbol{\varphi_{0}}+\boldsymbol{\varphi_{1}}\tau\omega,
\end{aligned}
\end{equation}
and represent \eqref{eq:block-iir} by $\mathbf{v}_i\le\boldsymbol{\mathcal{H}}(\omega)\mathbf{v}_{i-1}$, thus indicating that convergence holds if the spectral radius $\rho(\boldsymbol{\mathcal{H}}(\omega))<1$. However, the complexity of $\boldsymbol{\mathcal{H}}(\omega)$ makes $\rho(\boldsymbol{\mathcal{H}}(\omega))$ less explicit, and often challenging to derive since it departs from classical, non-recursive small-gain framework. Consequently, the standard analytical tools do not apply in the present setting. Alternatively, we adopt matrix perturbation theory \cite{kato2013perturbation} for our analysis as follows.

\begin{theorem}[Simple eigenvalue perturbation, restated in our notation \cite{kato2013perturbation,stewart1998perturbation,greenbaum2020first}]\label{lem:MATRIX}
Let $\boldsymbol{\mathcal{H}}(\omega)$ be analytic around $\omega_0$, $\boldsymbol{\mathcal{H}}(\omega_0)$ has a simple eigenvalue $\lambda_0$ with left/right eigenvectors $\boldsymbol{l}^*_{0},\boldsymbol{r}_{0}$ normalized by $\boldsymbol{l}^*_{0}\boldsymbol{r}_{0}=1$, then
there exists analytic maps $\lambda(\omega)$, $\boldsymbol{l}^*(\omega)$, $\boldsymbol{r}(\omega)$ that establish $\boldsymbol{l}^*(\omega)\boldsymbol{\mathcal{H}}(\omega)=\lambda(\omega)\boldsymbol{l}^*(\omega)$, $\boldsymbol{\mathcal{H}}(\omega)\boldsymbol{r}(\omega)=\lambda(\omega)\boldsymbol{r}(\omega)$ and $\boldsymbol{l}^*(\omega)\boldsymbol{r}(\omega)=1$ with initial condition $\lambda(\omega_0)=\lambda_0$, $\boldsymbol{l}^*(\omega_0)=\boldsymbol{l}^*_{0}$, $\boldsymbol{r}(\omega_0)=\boldsymbol{r}_{0}$. And the derivatives hold
\begin{equation}\label{eq:Th1-d}
\begin{aligned}
    &\lambda'(\omega) =\boldsymbol{l}^*(\omega) \boldsymbol{\mathcal{H}}'(\omega)\boldsymbol{r}(\omega),\\
    & \lambda''(\omega) =\boldsymbol{l}^*(\omega) \boldsymbol{\mathcal{H}}''(\omega)\boldsymbol{r}(\omega)-2\boldsymbol{l}^*(\omega)\boldsymbol{\mathcal{H}}'(\omega)\boldsymbol{\mathcal{S}}(\omega)\boldsymbol{\mathcal{H}}'(\omega)\boldsymbol{r}(\omega),\\
    &\lambda'''(\omega) \\ & =\boldsymbol{l}^*(\omega) \boldsymbol{\mathcal{H}}'''(\omega)\boldsymbol{r}(\omega) \\ & -3\boldsymbol{l}^*(\omega)\Big(\boldsymbol{\mathcal{H}}''(\omega)\boldsymbol{\mathcal{S}}(\omega)\boldsymbol{\mathcal{H}}'(\omega) +\boldsymbol{\mathcal{H}}'(\omega)\boldsymbol{\mathcal{S}}(\omega)\boldsymbol{\mathcal{H}}''(\omega)\Big)\boldsymbol{r}(\omega) 
 \\ &+ 6\boldsymbol{l}^*(\omega)\boldsymbol{\mathcal{H}}'(\omega)\boldsymbol{\mathcal{S}}(\omega)\Big(\boldsymbol{\mathcal{H}}'(\omega)-\lambda'(\omega)\mathbf{I}\Big)\boldsymbol{\mathcal{S}}(\omega) \boldsymbol{\mathcal{H}}'(\omega)\boldsymbol{r}(\omega),
\end{aligned}
\end{equation}
where $\boldsymbol{\mathcal{S}}(\omega)$ satisfies 
\begin{equation}\label{eq:Th1-d1}
\begin{aligned}
    & (\boldsymbol{\mathcal{H}}(\omega)-\lambda(\omega)\mathbf{I})\boldsymbol{\mathcal{S}}(\omega)=\boldsymbol{\mathcal{S}}(\omega)(\boldsymbol{\mathcal{H}}(\omega)-\lambda(\omega)\mathbf{I})=\mathbf{I}-\boldsymbol{\mathcal{P}}(\omega),\\
    & \boldsymbol{\mathcal{P}}(\omega)= \boldsymbol{r}(\omega) \boldsymbol{l}^*(\omega), \quad \boldsymbol{\mathcal{S}}(\omega)\boldsymbol{r}(\omega)=0, \quad \boldsymbol{l}^*(\omega)\boldsymbol{\mathcal{S}}(\omega)=0.
\end{aligned}
\end{equation}
\end{theorem}

\begin{definition}[Riesz projection and reduced resolvent \cite{greenbaum2020first,riesz2012functional}]\label{lem:dddd}
Let \(\Gamma\) be a simple closed contour in the complex plane that encloses the eigenvalue \(\lambda(\omega)\) of \(\boldsymbol{\mathcal{H}}(\omega)\) and no other eigenvalues. Then the operators \(\boldsymbol{\mathcal{P}}(\omega)\) and \(\boldsymbol{\mathcal{S}}(\omega)\) from \eqref{eq:Th1-d1} in Theorem~\ref{lem:MATRIX} can be equivalently written as
\begin{equation}\label{eq:d2-d1}
\begin{aligned}
\boldsymbol{\mathcal{P}}(\omega) & = -\frac{1}{2\pi j}\oint_{\Gamma}(\mathcal{H}(\omega)-z\mathbf{I})^{-1}dz, \\
\boldsymbol{\mathcal{S}}(\omega) & = \frac{1}{2\pi j}\oint_{\Gamma}\frac{1}{\lambda(\omega)-z}(\mathcal{H}(\omega)-z\mathbf{I})^{-1}dz,
\end{aligned}
\end{equation}
where \(\boldsymbol{\mathcal{P}}(\omega)\) is the Riesz projection associated with \(\lambda(\omega)\) and \(\boldsymbol{\mathcal{S}}(\omega)\) is the corresponding reduced resolvent.
\end{definition}

Theorem~\ref{lem:MATRIX} restates the classical Stewart–Sun/Kato results, giving explicit derivatives for a simple eigenvalue $\lambda(\omega)$ of an analytic matrix function $\boldsymbol{\mathcal{H}}(\omega)$ regardless of its complexity. Definition~\ref{lem:dddd} complements this by expressing the key operators, the Riesz projection \(\boldsymbol{\mathcal{P}}(\omega)\) and reduced resolvent \(\boldsymbol{\mathcal{S}}(\omega)\), as contour integrals. These analytic representations are crucial, as they allow us to use contour integration techniques to bound the norms of these operators, which in turn enables the derivation of computable bounds on the motion of the spectral radius $\rho(\boldsymbol{\mathcal{H}}(\omega))$.

Furthermore, we derive constants bounding the relative norms of $\{\boldsymbol{\psi_{0}},\boldsymbol{\psi_{1}},\boldsymbol{\varphi_{0}},\boldsymbol{\varphi_{1}}\}$ in \eqref{eq:block-all}
\begin{equation}\label{eq:constants}
\begin{aligned}
c_1 & = 2-1/\tau \ge \Vert (\mathbf{I}-\boldsymbol{\psi_{0}})^{-1} \Vert_{\infty},
\\ c_2 &=4\tau(\tau-1) \ge \Vert (\mathbf{I}-\boldsymbol{\psi_{0}})^{-1}\Vert_{\infty} \Vert\boldsymbol{\psi_{1}}(\tau-1)\Vert_{\infty},\\
c_3 &= 1+\epsilon_\tau \ge \Vert \boldsymbol{\varphi_{0}} \Vert_{\infty}, \quad c_4 = 2\tau+1\ge\Vert \boldsymbol{\varphi_{1}} \Vert_{\infty};
\end{aligned}
\end{equation}
together with composite constants
\begin{equation}\label{eq:constants-2}
\begin{aligned}
 c_{\boldsymbol{\mathcal{P}}} & = 2+4\frac{1+\epsilon_{\tau}}{1-\epsilon_{\tau}},\quad c_{\boldsymbol{\mathcal{S}}} =\frac{4}{1-\epsilon_{\tau}}c_{\boldsymbol{\mathcal{P}}}, \\
 c_{\boldsymbol{\mathcal{H}}'}& =2c_1(c_2c_3+\tau c_4) \Big(\frac{2}{1-\epsilon_{\tau}}+\frac{4(1+\epsilon_\tau)}{(1-\epsilon_{\tau})^2}\Big), \\
K_3 &= c_{\boldsymbol{\mathcal{P}}}\Big(19c_1 c^2_2(c_2c_3+\tau c_4)+60c_{\boldsymbol{\mathcal{S}}}c^2_1 c_2(c_2c_3+\tau c_4)^2\\
&+48c^2_{\boldsymbol{\mathcal{S}}} (1+c_{\boldsymbol{\mathcal{P}}})c^3_1(c_2c_3+\tau c_4)^3\Big),
\end{aligned}
\end{equation}
whose precise roles will be given in the subsequent section, as preparation for the convergence result that follows.

\begin{theorem}[The linear convergence of: \textbf{\texttt{GraT-Diff}}]\label{thm:global-sg}
Let Assumptions \ref{as:smooth-strong} and \ref{as:stochastic} hold, together with constant sets \eqref{eq:constants} \eqref{eq:constants-2}. If the stepsize $\mu$ in \textnormal{\textbf{\texttt{GraT-Diff}}} satisfies
\begin{equation}\label{eq:mu-sg-explicit}
\begin{aligned}
    \mu\ \le \min \bigg\{ & \frac{1}{\tau m},\ \frac{1}{4c_2L},\ \frac{1}{2c_{\boldsymbol{\mathcal{H}}'}L},\
\frac{1}{L}\sqrt\frac{3\tau}{2QK_3}
\bigg\}.
\end{aligned}
\end{equation}
then there exists $\varrho>0$ such that
\begin{equation}\label{eq:block-linear}
C_i + T_i + S_i \ \le\ \varrho\,\Big(1-\frac{m\tau}{4}\mu\Big)^{\,i}\,(C_0+T_0+S_0) \quad \text{for all} \,\,i\ge 0.
\end{equation}
Consequently, $\|\mathcal{C}^k\|\to 0$, $\|\mathcal{T}^k\|\to 0$, and $\|\mathbf{s}^k\|\to 0$, the global state $\mathcal{Z}^{k}$ converges to the exact minimizer $\mathbf{1}\!\otimes\!\mathbf{z}^\star$ linearly.
\end{theorem}

\begin{figure}[t]
  \centering
  \includegraphics[width=0.8\linewidth]{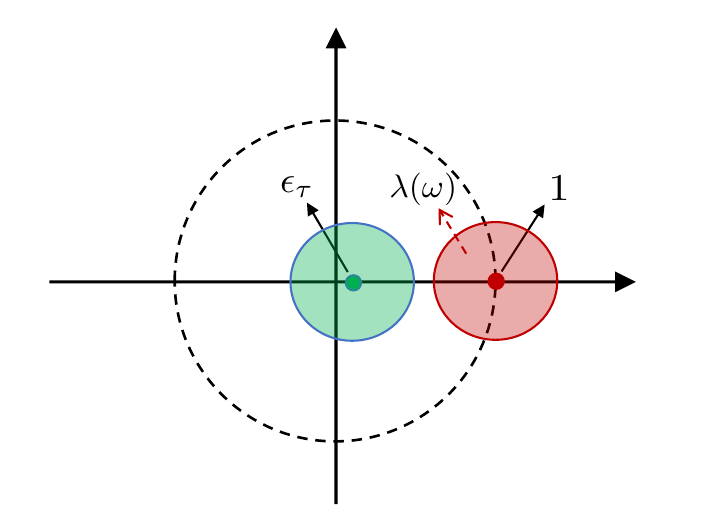}
    \caption{Spectrum of \(\boldsymbol{\mathcal{H}}(\omega)\). 
The green disk bounds the non-dominant modes associated with the eigenvalue \(\epsilon_\tau\), whereas the red disk isolates the dominant eigenvalue branch \(\lambda(\omega)\) perturbing from the eigenvalue \(1\), which determines the spectral radius. Theorem~\ref{lem:MATRIX} helps with the analysis, which reduces to controlling the Taylor expansion $\lambda(\omega) = 1-\frac{\tau}{2Q}\omega + \frac{\lambda'''(\xi)}{6}\omega^{3}, \quad \text{for} \quad 0<\xi \le \omega$. The key step is to absorb the remainder term \(\frac{\lambda'''(\xi)}{6}\omega^{3}\) so that the leading linear decay yields linear convergence.}
  \label{fig:spectrum}
\end{figure}

\begin{proof}
The main idea of the proof is illustrated in Fig.~\ref{fig:spectrum}. 
We expand $\boldsymbol{\mathcal{H}}(\omega)$ around $\omega_0=0$ to establish convergence for sufficiently small stepsizes. By having $\omega_0=0$ in \eqref{eq:block-all} \eqref{eq:block-iir}, we derive
\begin{equation}\label{eq:Ti-closed}
\begin{aligned}
   \boldsymbol{\mathcal{H}}(0)=  \begin{bmatrix}
\epsilon_{\tau} &
0&
0\\
1+(1-\frac{1}{\tau})\epsilon_{\tau} & \epsilon_{\tau} & 0\\
0 &
0 &
1
\end{bmatrix}.
\end{aligned}
\end{equation}
As $\boldsymbol{\mathcal{H}}(0)$ is a lower triangular matrix, we find $\{\epsilon_\tau,\epsilon_\tau,1\}$ form its eigenvalues, and $1$ is the simple eigenvalue that also dominates the spectral radius ($\epsilon_\tau <1 $). In addition, $1$'s corresponding eigenvector is as simple as $\mathbf{e}_3=[0,0,1]^\top$. By Theorem~\ref{lem:MATRIX}, perturbations $\lambda(\omega), \boldsymbol{l}^*(\omega),\boldsymbol{r}(\omega)$ with initialization $\lambda_0=1, \,\boldsymbol{l}_{0}=\boldsymbol{r}_{0}=\mathbf{e}_3$ can be analyzed as the characterization of spectral radius motions of $\boldsymbol{\mathcal{H}}(\omega)$. Specifically, we derive the point eigenvalue derivative
\begin{equation}\label{eq:Ti-eliminate}
\begin{aligned}
\lambda'(0)&=\boldsymbol{l}^*_0 \boldsymbol{\mathcal{H}}'(0)\boldsymbol{r}_0=\boldsymbol{e}^\top_{3}\boldsymbol{\mathcal{H}}'(0)\boldsymbol{e}_{3}=\frac{-\tau}{2Q},\\
\lambda''(0)&=\boldsymbol{l}^*_0 \boldsymbol{\mathcal{H}}''(0)\boldsymbol{r}_0-2\boldsymbol{l}^*_0\boldsymbol{\mathcal{H}}'(0)\boldsymbol{\mathcal{S}}(0)\boldsymbol{\mathcal{H}}'(0)\boldsymbol{r}_0=0,
\end{aligned}
\end{equation}
based on differentials of $\boldsymbol{\mathcal{H}}(\omega)$
\begin{equation}\label{eq:Ti-eliminate22}
\begin{aligned}
   \boldsymbol{\mathcal{H}}'(\omega) & =\boldsymbol{\mathcal{R}}(\omega)\big(\boldsymbol{\psi}'(\omega)\boldsymbol{\mathcal{R}}(\omega)\boldsymbol{\varphi}(\omega)+\boldsymbol{\varphi}'(\omega)\big),
   \\  \boldsymbol{\mathcal{H}}''(\omega)&=2\boldsymbol{\mathcal{R}}(\omega)\boldsymbol{\psi}'(\omega)\boldsymbol{\mathcal{R}}(\omega)\big(\boldsymbol{\psi}'(\omega)\boldsymbol{\mathcal{R}}(\omega)\boldsymbol{\varphi}(\omega) + \boldsymbol{\varphi}'(\omega)\big), \\ \boldsymbol{\mathcal{H}}'''(\omega)&=6\big(\boldsymbol{\mathcal{R}}(\omega)\boldsymbol{\psi}'(\omega)\big)^2\boldsymbol{\mathcal{R}}(\omega)\big(\boldsymbol{\psi}'(\omega)\boldsymbol{\mathcal{R}}(\omega)\boldsymbol{\varphi}(\omega) + \boldsymbol{\varphi}'(\omega)\big),
\end{aligned}
\end{equation}
at $\omega_0=0$ and the reduced resolvent \eqref{eq:Th1-d1}
\begin{equation}\label{eq:Ti-3334}
\begin{aligned}
   \boldsymbol{\mathcal{S}}(0)=  \begin{bmatrix}
\frac{1}{\epsilon_{\tau}-1} &
0&
0\\
\frac{-1-(1-\frac{1}{\tau})\epsilon_{\tau}}{(\epsilon_{\tau}-1)^2} & \frac{1}{\epsilon_{\tau}-1} & 0\\
0 &
0 &
0
\end{bmatrix}.
\end{aligned}
\end{equation}
Then, by Taylor expansion at $\omega_0=0$, it holds
\begin{equation}\label{eq:Ti-333}
\begin{aligned}
\lambda(\omega)& =\lambda(0)+ \lambda'(0)\omega+ \frac{\lambda''(0)}{2}\omega^2 + \frac{\lambda'''(\xi)}{6}\omega^{3}\\
 & =1-\frac{\tau}{2Q}\omega + \frac{\lambda'''(\xi)}{6}\omega^{3}, \quad \text{for} \quad 0<\xi \le \omega.
\end{aligned}
\end{equation}
As we can see, \eqref{eq:Ti-333} tells an unachievable maximum rate $1-\frac{\tau}{2Q}\omega$ with a remainder $\frac{\lambda'''(\xi)}{6}\omega^{3}$ if $\lambda(\omega)$ dominates the radius by $\vert\lambda(\omega)\vert=\rho(\boldsymbol{\mathcal{H}}(\omega))$. To bound $\lambda(\omega)$'s motion, we need the remainder to be absorbed, and it holds
\begin{equation}\label{eq:Timu}
\begin{aligned}
&|\lambda'''(\xi)| \\
&\le \bigg| \boldsymbol{l}^*(\xi) \boldsymbol{\mathcal{H}}'''(\xi)\boldsymbol{r}(\xi) \\ & -3\boldsymbol{l}^*(\xi)\Big(\boldsymbol{\mathcal{H}}''(\xi)\boldsymbol{\mathcal{S}}(\xi)\boldsymbol{\mathcal{H}}'(\xi) +\boldsymbol{\mathcal{H}}'(\xi)\boldsymbol{\mathcal{S}}(\xi)\boldsymbol{\mathcal{H}}''(\xi)\Big)\boldsymbol{r}(\xi)\\
 & + 6\boldsymbol{l}^*(\xi)\boldsymbol{\mathcal{H}}'(\xi)\boldsymbol{\mathcal{S}}(\xi)\Big(\boldsymbol{\mathcal{H}}'(\xi)-\lambda'(\xi)\mathbf{I}\Big)\boldsymbol{\mathcal{S}}(\xi) \boldsymbol{\mathcal{H}}'(\xi)\boldsymbol{r}(\xi) \bigg| \\
 &\le \Vert\boldsymbol{\mathcal{P}}(\xi)\Vert_{\infty}\Big(\Vert\boldsymbol{\mathcal{H}}'''(\xi)\Vert_{\infty}
 \\ &+6\Vert \boldsymbol{\mathcal{H}}''(\xi)\Vert_{\infty} \Vert \boldsymbol{\mathcal{S}}(\xi)\Vert_{\infty}\Vert\boldsymbol{\mathcal{H}}'(\xi)\Vert_{\infty} \\
& +6(1+\Vert\boldsymbol{\mathcal{P}}(\xi)\Vert_{\infty})\Vert \boldsymbol{\mathcal{S}}(\xi)\Vert_{\infty}^2 \Vert\boldsymbol{\mathcal{H}}'(\xi)\Vert_{\infty}^3 \Big),
\end{aligned}
\end{equation}
where the second inequality is based on $|\boldsymbol{l}^*(\xi)(\cdot)\boldsymbol{r}(\xi)|\le\Vert\boldsymbol{\mathcal{P}}(\xi)\Vert_{\infty} \Vert(\cdot) \Vert_{\infty}$ and $\lambda'(\xi) =\boldsymbol{l}^*(\xi) \boldsymbol{\mathcal{H}}'(\xi)\boldsymbol{r}(\xi)$. As we can see, several quantities ($\Vert \boldsymbol{\mathcal{H}}'(\xi)\Vert_{\infty},\Vert \boldsymbol{\mathcal{H}}''(\xi)\Vert_{\infty},\Vert \boldsymbol{\mathcal{H}}'''(\xi)\Vert_{\infty}$, $\Vert\boldsymbol{\mathcal{P}}(\xi)\Vert_{\infty}$ and $\Vert \boldsymbol{\mathcal{S}}(\xi)\Vert_{\infty}$) are needed for \eqref{eq:Timu}. From the relations in \eqref{eq:block-continuous} and \eqref{eq:Ti-eliminate22}, one can derive $\Vert \boldsymbol{\mathcal{R}}(\xi)\Vert_{\infty}$ as preparation for $\Vert \boldsymbol{\mathcal{H}}'(\xi)\Vert_{\infty},\Vert \boldsymbol{\mathcal{H}}''(\xi)\Vert_{\infty},\Vert \boldsymbol{\mathcal{H}}'''(\xi)\Vert_{\infty}$. 
Let $\xi \le \omega \le \frac{1}{4c_2}$ (\textit{the second term in \eqref{eq:mu-sg-explicit} appears}), $\mathbf{I}-\boldsymbol{\psi}(\omega)$ is invertible and $\boldsymbol{\mathcal{R}}(\xi)$ can be bounded
\begin{equation}\label{eq:TiM_renew}
\begin{aligned}
&\Vert \boldsymbol{\mathcal{R}}(\xi)\Vert_{\infty} \\ &\le  \underbrace{\Vert (\mathbf{I}-\boldsymbol{\psi_{0}})^{-1} \Vert_{\infty}}_{\le c_1} \Vert(\mathbf{I}-(\mathbf{I}-\boldsymbol{\psi_{0}})^{-1} \boldsymbol{\psi_{1}}(\tau-1)\xi)^{-1}\Vert_{\infty}  \\
&\le c_1 \frac{1}{1-\underbrace{\Vert(\mathbf{I}-\boldsymbol{\psi_{0}})^{-1} \boldsymbol{\psi_{1}}(\tau-1)\Vert_{\infty}}_{ \le c_2}\xi}\le c_1\frac{1}{1-c_2 \xi} \le \frac{4}{3}c_1
\end{aligned}
\end{equation}
where the second inequality is based on the Neumann series sum over $\xi$, and consolidating constants via \eqref{eq:constants}. Under \eqref{eq:TiM_renew}, it further establishes
\begin{equation}\label{eq:Ti-HHHH}
\begin{aligned}
  & \Vert \boldsymbol{\mathcal{H}}'(\xi)\Vert_{\infty} \le \Vert \boldsymbol{\mathcal{R}}(\xi)\Vert_{\infty}\big(\underbrace{\Vert\big(\boldsymbol{\psi}'(\xi)\boldsymbol{\mathcal{R}}(\xi)\Vert_{\infty}}_{\le c_2/(1-c_2 \xi)}\Vert\boldsymbol{\varphi}(\xi)\Vert_{\infty} +\Vert\boldsymbol{\varphi}'(\xi)\Vert_{\infty}\big) \\ &\le \frac{c_1}{(1-c_2 \xi)^2}(c_2c_3+\tau c_4) \le2c_1(c_2c_3+\tau c_4), \\
 & \Vert \boldsymbol{\mathcal{H}}''(\xi)\Vert_{\infty} \le 2 \underbrace{\Vert\boldsymbol{\mathcal{R}}(\xi)\boldsymbol{\psi}'(\xi)  \Vert_{\infty}}_{\le c_2/(1-c_2 \xi)} \Vert \boldsymbol{\mathcal{H}}'(\xi)\Vert_{\infty}  \\
 &\le  \frac{2c_1c_2}{(1-c_2 \xi)^3}(c_2c_3+\tau c_4) \le 5c_1 c_2(c_2 c_3+\tau c_4), \\
  &  \Vert \boldsymbol{\mathcal{H}}'''(\xi)\Vert_{\infty} \le  6 \underbrace{\Vert\boldsymbol{\mathcal{R}}(\xi)\boldsymbol{\psi}'(\xi)  \Vert^2_{\infty}}_{\le c^2_2/(1-c_2 \xi)^2} \Vert \boldsymbol{\mathcal{H}}'(\xi)\Vert_{\infty}\\
  & \le \frac{6c_1 c^2_2}{(1-c_2 \xi)^4}(c_2 c_3+\tau c_4)\le19c_1 c^2_2(c_2 c_3+\tau c_4),
\end{aligned}
\end{equation}
by expanding \eqref{eq:Ti-eliminate22}, invoking the norm’s subadditivity and submultiplicativity, and consolidating constants similar to \eqref{eq:TiM_renew}. Subsequently, it remains to derive $\Vert\boldsymbol{\mathcal{P}}(\xi)\Vert_{\infty}$ and $\Vert \boldsymbol{\mathcal{S}}(\xi)\Vert_{\infty}$ for the final bound on \eqref{eq:Timu}. Let $g=1-\epsilon_\tau $ be spectral gap of eigenvalue set $\{\epsilon_\tau,\epsilon_\tau,1\}$ from $\boldsymbol{\mathcal{H}}(0)$ and set a contour $\Gamma: |z-1|=g/2$, the resolvent
\begin{equation}\label{eq:Ti-ssssH}
\begin{aligned}
  &\sup_{z\in\Gamma} \Vert(\boldsymbol{\mathcal{H}}(0)-z\mathbf{I})^{-1}\Vert_{\infty} \\
  &= \sup_{z\in\Gamma} \Bigg\Vert \begin{bmatrix}
\frac{1}{z-\epsilon_{\tau}} &
0&
0\\
\frac{-1-(1-\frac{1}{\tau})\epsilon_{\tau}}{(z-\epsilon_{\tau})^2} & \frac{1}{z-\epsilon_{\tau}} & 0\\
0 &
0 &
\frac{1}{z-1}
\end{bmatrix}\Bigg\Vert_{\infty}\le \frac{2}{g}+\frac{4(1+\epsilon_\tau)}{g^2}.
\end{aligned}
\end{equation}
By introducing $\Delta(\xi)=\boldsymbol{\mathcal{H}}(\xi)-\boldsymbol{\mathcal{H}}(0)$, it holds
\begin{equation}\label{eq:Ti-ssssnew}
\begin{aligned}
  &\sup_{z\in\Gamma}  \Vert(\boldsymbol{\mathcal{H}}(\xi)-z\mathbf{I})^{-1}\Vert_{\infty} \\
  & =  \sup_{z\in\Gamma}  \Big\Vert \big(\mathbf{I}+(\boldsymbol{\mathcal{H}}(0)-z\mathbf{I})^{-1}\Delta(\xi)\big)^{-1} (\boldsymbol{\mathcal{H}}(0)-z\mathbf{I})^{-1}\Big\Vert_{\infty}.
\end{aligned}
\end{equation}
Let $\xi \le \omega\le \frac{1}{2c_{\boldsymbol{\mathcal{H}}'}}$ (\textit{the third term in \eqref{eq:mu-sg-explicit} appears}), Neumann smallness condition holds:
\begin{equation}\label{eq:Ti-ssssneww2}
\begin{aligned}
 &\sup_{z\in\Gamma}  \big\Vert (\boldsymbol{\mathcal{H}}(0)-z\mathbf{I})^{-1}\big\Vert_{\infty}\big\Vert \Delta(\xi) \big\Vert_{\infty} \\
 &\le\Big(\frac{2}{g}+\frac{4(1+\epsilon_\tau)}{g^2} \Big) \Vert \boldsymbol{\mathcal{H}}'(\xi)\Vert_{\infty} \frac{1}{2c_{\boldsymbol{\mathcal{H}}'}} \le \frac{1}{2},
\end{aligned}
\end{equation}
where the first inequality is based on the mean value theorem on $\Delta(\xi)$ and \eqref{eq:constants-2} \eqref{eq:Ti-HHHH} \eqref{eq:Ti-ssssH}. Here, $c_{\boldsymbol{\mathcal{H}}'}$ appears as the separation guarantee for the spectrum. With the Neumann series
sum condition, \eqref{eq:Ti-ssssnew} can thus be rewritten as
\begin{equation}\label{eq:Ti-ssssnew2}
\begin{aligned}
 \sup_{z\in\Gamma}  \Vert(\boldsymbol{\mathcal{H}}(\xi)-z\mathbf{I})^{-1}\Vert_{\infty} \le2\Big(\frac{2}{g}+\frac{4(1+\epsilon_\tau)}{g^2} \Big).
\end{aligned}
\end{equation}
By \eqref{eq:Ti-ssssnew2}, the resolvent is well-defined on $\Gamma$, so exactly one eigenvalue \(\lambda(\omega)\) of \(\boldsymbol{\mathcal H}(\omega)\) lies in the disk of $\Gamma$.
One can check \(\Gamma_\varepsilon :|z-\epsilon_\tau|=g/2\) traps the other two eigenvalues within the $\Gamma_\varepsilon$ under the same condition. Hence \(|\lambda(\omega)|=\rho(\boldsymbol{\mathcal H}(\omega))\), which justifies the initial statement. By Definition~\ref{lem:dddd}, we have
\begin{equation}\label{eq:Ti-cps}
\begin{aligned}
&\Vert\boldsymbol{\mathcal{P}}(\xi)\Vert_{\infty} \\ & \le \frac{1}{2\pi}\int_{\Gamma}\Vert(\mathcal{H}(\xi)-z\mathbf{I})^{-1}\Vert_{\infty} |dz| \\ & \le \frac{g}{2} \cdot \sup_{z\in\Gamma}  \Vert(\boldsymbol{\mathcal{H}}(\xi)-z\mathbf{I})^{-1}\Vert_{\infty} \le2+4\frac{1+\epsilon_{\tau}}{1-\epsilon_{\tau}}=c_{\boldsymbol{\mathcal{P}}},
\end{aligned}
\end{equation}
\begin{equation}\label{eq:Ti-cpsP}
\begin{aligned}
&\Vert\boldsymbol{\mathcal{S}}(\xi)\Vert_{\infty} \\ & \le \frac{1}{2\pi}\int_{\Gamma}\frac{1}{|\lambda(\xi)-z|}\Vert(\mathcal{H}(\xi)-z\mathbf{I})^{-1}\Vert_{\infty} |dz|   \\
& \le \frac{g}{2} \cdot \frac{4}{g}\cdot\sup_{z\in\Gamma}  \Vert(\boldsymbol{\mathcal{H}}(\xi)-z\mathbf{I})^{-1}\Vert_{\infty} \le \frac{4}{1-\epsilon_{\tau}}c_{\boldsymbol{\mathcal{P}}} = c_{\boldsymbol{\mathcal{S}}},
\end{aligned}
\end{equation}
where we assume $|\lambda(\xi)-1|\le g/4$ and $|\lambda(\xi)-z|\ge g/4$. Here, $c_{\boldsymbol{\mathcal{P}}}$ and $c_{\boldsymbol{\mathcal{S}}}$ serve as the upper bounds of operators echoed with \eqref{eq:constants-2}.
The assumption is justified by shrinking $\xi$ by a constant so the Neumann condition holds on $\Gamma_{1/4}:|z-1|=g/4$; the projection bound $\Vert\boldsymbol{\mathcal{P}}(\xi)\Vert_{\infty}$ on $\Gamma$ is unchanged, only the admissible step size scales. Combining with \eqref{eq:Ti-HHHH} \eqref{eq:Ti-cps} \eqref{eq:Ti-cpsP}, we can bound \eqref{eq:Timu} by
\begin{equation}\label{eq:Ti-cpsPok}
\begin{aligned}
|\lambda'''(\xi)| \le K_3.
\end{aligned}
\end{equation}
Let $\xi \le \omega\le \sqrt\frac{3\tau}{2QK_3}$ (\textit{the last term in \eqref{eq:mu-sg-explicit} appears}), then
\begin{equation}\label{eq:Ti-cpsPokss}
\begin{aligned}
 \frac{|\lambda'''(\xi)|}{6}\omega^{2}\le \frac{\tau}{4Q}.
\end{aligned}
\end{equation}
This leads the remainder to be absorbed, and the final rate becomes linear by $1-\frac{\tau}{4Q}\omega$. The proof is complete. 
\end{proof}

\begin{remark}[Trade-offs and choices]\label{rm:4}
Because $\lambda''(0)=0$, the first nontrivial stepsize restriction comes from the third-order Taylor term, controlled via $|\lambda'''(\xi)|$, which typically yields a sharper bound. For a more conservative guarantee, one may instead use an expansion with a second-order remainder. The choice of contour radius around $\Gamma$ trades off the constants: $\tfrac12$ is a simple, roughly balanced option, but tuning this ratio can reduce conservatism. Using other induced norms (e.g., spectral or weighted $\ell_\infty$) can tighten the constant sets in \eqref{eq:constants} and \eqref{eq:constants-2} and thereby enlarge the admissible range of $\mu$. The bound can also be simplified to an $\mathcal{O}(\cdot)$ form by further absorbing terms in these constant sets, which we leave for future refinement. In practice, gaps are often larger and couplings smaller than in the worst case, so stable simulations typically admit stepsizes exceeding the conservative threshold in \eqref{eq:mu-sg-explicit}. Moreover, for a fixed stepsize, decreasing $\epsilon_\tau$ can further accelerate convergence.
\end{remark}

\appendices
\renewcommand{\thesection}{\arabic{section}}
\renewcommand{\thesectiondis}{S\arabic{section}}
\section{Proof of Lemma 1}
\label{Note:A}
Multiplying both sides of (10) by $\tfrac{1}{R}\big(\mathbf{1}^{\top}\!\otimes\!\mathbf{I}_d\big)$ and using (12), it shows
\begin{equation}
\begin{aligned}
\bar{\mathbf{z}}^{k+1} &= \tfrac{1}{R}\big(\mathbf{1}^{\top}\!\otimes\!\mathbf{I}_d\big)\mathcal{W}^k\big(\mathcal{Z}^{k}-\mu\,\mathit{G}^{k}\big),\\
\bar{\mathbf{g}}^{k+1} &= \tfrac{1}{R}\big(\mathbf{1}^{\top}\!\otimes\!\mathbf{I}_d\big)\mathcal{W}^k\,\mathit{G}^{k} + \overline{\nabla f}^{k+1} -\overline{\nabla f}^k.
\end{aligned}
\label{eq:l11}
\end{equation} 
Using the mixed product property on $\mathcal{W}^k=\mathbf{W}^{(k\%\tau)}\otimes \mathbf{I}_d$ and $\mathbf{1}^{\mathsf T}\mathbf{W}^{(k\%\tau)}=\mathbf{1}^{\mathsf T}$ in Assumption~\ref{as:stochastic} , it holds
\begin{equation}
    \begin{aligned}
        \tfrac{1}{R}\big(\mathbf{1}^{\top}\!\otimes\!\mathbf{I}_d\big)\mathcal{W}^k= \tfrac{1}{R}\big(\mathbf{1}^{\top}\!\otimes\!\mathbf{I}_d\big),
    \end{aligned}
\end{equation}
Then ~\eqref{eq:avg-z} holds.
Similarly, given $\mathbf{g}_r^0=\nabla f_r(\mathbf{z}_r^0)$, telescoping \eqref{eq:l11} gives $ \bar{\mathbf{g}}^{k}=\overline{\nabla f}^k$ and hence~\eqref{eq:avg-g}. According to the definition of $\Delta^k$ in \eqref{notations}, the equality in \eqref{eq:delta-bound} holds. Since each $f_r(\mathbf{z})$ is 
$L$-smooth by Assumption~\ref{as:smooth-strong}, we can write 
\begin{equation}
    \begin{aligned}
        \left\|\frac1R\sum_{r=1}^R \big(\nabla f_r(\mathbf{z}_r^{k}) - \nabla f_r(\bar{\mathbf{z}}^{k})\big)\right\|\le \frac{L}{R} \sum_{r=1}^R\|\mathbf{z}_r^k-\bar{\mathbf{z}}^k\|.
    \end{aligned}
    \label{eq:l12}
\end{equation}
By Cauchy–Schwarz inequality, we further bounded
\begin{equation}
    \begin{aligned}
\frac{L}{R} \sum_{r=1}^R\|\mathbf{z}_r^k-\bar{\mathbf{z}}^k\| \le \frac{L(\sum_{r=1}^R 1)^{1/2}}{R} \Big(\sum_{r=1}^R \|\mathbf{z}_r^k-\bar{\mathbf{z}}^k\|^{2}\Big)^{1/2}
    \end{aligned}
    \label{eq:l13}
\end{equation}
of which the r.h.s is $\frac{L}{\sqrt{R}}\,\|\mathcal{C}^k\|$. The proof is completed.\hfill$\square$
\section{Proof of Lemma 2}
\label{Note:B}
Merging ~\eqref{eq:avg-z}-\eqref{eq:avg-g}, we can write
\begin{equation}
    \begin{aligned}
\bar{\mathbf{z}}^{k+1}=\bar{\mathbf{z}}^{k}-\mu\nabla F(\bar{\mathbf{z}}^{k})-\mu \Delta^k.
    \end{aligned}
    \label{eq:l21}
\end{equation}
Subtracting $\mathbf{z}^\star$ from \eqref{eq:l21} and inserting $\nabla F(\mathbf{z}^\star)=0$ into the r.h.s., it holds
\begin{equation}
    \begin{aligned}
\mathbf{s}^{k+1} 
=(\mathbf{I}_d-\mu  \mathbf{H})\mathbf{s}^{k}-\mu\Delta^k,
    \end{aligned}
    \label{eq:l22}
\end{equation}
Under Assumption~\ref{as:smooth-strong}, $F(\mathbf{z})$ is $m$-strongly convex and $L$-smooth ($m\mathbf{I}_d\preceq \nabla^2 F(\mathbf{z})\preceq L\mathbf{I}_d$), and hence for $\mu\le 1/L$,
\begin{equation}
(1-\mu L)\mathbf{I}_d\preceq \mathbf{I}_d-\mu  \mathbf{H}\preceq (1-\mu m)\mathbf{I}_d.
\label{eq:l2ine}
\end{equation}
Using \eqref{eq:l22} and the triangle inequality,
\begin{equation}\label{eq:exact-final}
\begin{aligned}
   & \|\mathbf{s}^{k+1}\| 
=\big\|(\mathbf{I}_d-\mu \mathbf{H})\mathbf{s}^{k}-\mu\Delta^k\|\\&
\le \big\|(\mathbf{I}_d-\mu \mathbf{H})\mathbf{s}^{k}\big\| + \mu\|\Delta^k\|
\overset{\eqref{eq:l2ine}}{\le} (1-\mu m)\,\|\mathbf{s}^{k}\|+\mu\|\Delta^k\|.
\end{aligned}
\end{equation}
Finally, by \eqref{eq:delta-bound}, $\|\Delta^k\|\le \tfrac{L}{\sqrt{R}}\|\mathcal{C}^k\|$, which yields the stated second inequality.
\hfill$\square$

\section{Proof of Lemma 3}
\label{Note:C}
Multiplying $\mathcal{J}_\perp$ on both sides of the first line in (10), using $\mathcal{J}_\perp \mathcal{W}^k=\!(\mathcal{W}^k-\mathcal{J})$ and $\mathcal{J}_\perp \mathcal{J}=\mathbf{0}$, we obtain 
\begin{equation}\label{eq:L31}
\begin{aligned}
\mathcal{C}^{k+1}&=\mathcal{J}_\perp \mathcal{Z}^{k+1}=(\mathcal{W}^k-\mathcal{J})(\mathcal{Z}^{k}-\mu\mathit{G}^{k})\\
&=(\mathcal{W}^k-\mathcal{J})(\mathcal{C}^{k}-\mu\mathcal{T}^{k}).
\end{aligned}
\end{equation}
Iterating over $\{(i-1)\tau,\cdots,k-1,k\}$, it follows that
\begin{equation}\label{eq:L32}
\begin{aligned}
\mathcal{C}^{k}&= \Big(\prod_{j=i\tau}^{k-1} (\mathcal{W}^j-\mathcal{J})\Big)\Big(\prod_{j=(i-1)\tau}^{i\tau-1} (\mathcal{W}^j-\mathcal{J})\Big)\mathcal{C}^{(i-1)\tau}\\
&-\mu\sum_{j=(i-1)\tau}^{k-1}\Big(\prod_{t=j}^{k-1} (\mathcal{W}^t-\mathcal{J})\Big)\mathcal{T}^{j}.
\end{aligned}
\end{equation}
where $\prod$ is the left-wise product. Taking the spectral norm on \eqref{eq:L32}, using norm submultiplicativity together with 
\begin{equation}\label{eq:L33}
\begin{aligned}
&\big\Vert\mathcal{W}^j-\mathcal{J}\big\Vert=\big\Vert (\mathbf{W}^{(j\%\tau)}-\tfrac{1}{R}\,\mathbf{1}\mathbf{1}^{\!\top})\otimes \mathbf{I}_d\Vert\leq 1\\
&\Big\Vert\prod_{j=(i-1)\tau}^{i\tau-1} (\mathcal{W}^j-\mathcal{J})\Big\Vert=\Big\Vert\big(\prod_{j=0}^{\tau-1} \mathbf{W}^j-\tfrac{1}{R}\,\mathbf{1}\mathbf{1}^{\!\top}\big)\otimes \mathbf{I}_d\Big\Vert=\epsilon_\tau
\end{aligned}
\end{equation}
from Assumption~\ref{as:stochastic} yields~\eqref{eq:block-e}.
\hfill$\square$

\section{Proof of Lemma 4}
\label{Note:D}
Similar to \eqref{eq:L31}, we can write 
\begin{equation}\label{eq:L41}
\begin{aligned}
\mathcal{T}^{k+1}
= (\mathcal{W}^k-\mathcal{J})\mathcal{T}^{k}
+ \mathcal{J}_\perp\!\big(\nabla \mathcal{F}^{\,k+1}-\nabla \mathcal{F}^{\,k}\big).
\end{aligned}
\end{equation}
Iterating over $\{(i-1)\tau,\cdots,k-1,k\}$, and using $\mathcal{J}_\perp \mathcal{W}^k=\mathcal{W}^k\mathcal{J}_\perp $ and $\mathcal{J}_\perp \mathcal{J}=\mathbf{0}$, it follows that
\begin{equation}\label{eq:L42}
\begin{aligned}
\mathcal{T}^{k}&= \Big(\prod_{j=i\tau}^{k-1} (\mathcal{W}^j-\mathcal{J})\Big)\Big(\prod_{j=(i-1)\tau}^{i\tau-1} (\mathcal{W}^j-\mathcal{J})\Big)\mathcal{T}^{(i-1)\tau}\\
&+\sum_{j=(i-1)\tau}^{k-1}\Big(\prod_{t=j+1}^{k-1} (\mathcal{W}^t-\mathcal{J})\Big)\big(\nabla \mathcal{F}^{\,j+1}-\nabla \mathcal{F}^{\,j}\big).
\end{aligned}
\end{equation}
Taking the spectral norm and using \eqref{eq:L33}, we write
\begin{equation}\label{eq:L43}
\begin{aligned}
\Vert\mathcal{T}^{k}\Vert\leq \epsilon_\tau \Vert\mathcal{T}^{(i-1)\tau}\Vert
+\sum_{j=(i-1)\tau}^{k-1}\Vert\nabla \mathcal{F}^{\,j+1}-\nabla \mathcal{F}^{\,j}\Vert.
\end{aligned}
\end{equation}
By $L$-smoothness,
\begin{equation}\label{eq:L44}
\big\Vert\nabla \mathcal{F}^{\,j+1}-\nabla \mathcal{F}^{\,j}\big\Vert
\le L\,\|\mathcal{Z}^{j+1}-\mathcal{Z}^{j}\|.
\end{equation}
By (10),
$\mathcal{Z}^{j+1}-\mathcal{Z}^{j}=(\mathcal{W}^j-\mathbf{I}_{Rd})\,\mathcal{Z}^{j}-\mu\,\mathcal{W}^j\,\mathit{G}^j$, and we can decompose $\mathcal{Z}^{j}=\mathcal{C}^{j}+\mathbf{1}\!\otimes\!\bar{\mathbf{z}}^{j}$ with $(\mathcal{W}^j-\mathbf{I}_{Rd})(\mathbf{1}\!\otimes\!\bar{\mathbf{z}}^{j})=\mathbf{0}$, it therefore holds
\begin{equation}\label{eq:L45}
\begin{aligned}
\|\mathcal{Z}^{j+1}-\mathcal{Z}^{j}\|
\le \underbrace{\|\mathcal{W}^j-\mathbf{I}_{Rd}\|_{\mathcal{J}_\perp}}_{\le\,c_W\,\le 2}\,\|\mathcal{C}^j\|
+ \mu\,\underbrace{\|\mathcal{W}^j\|}_{\le 1}\,\|\mathit{G}^j\|.
\end{aligned}
\end{equation}
Decomposing $\mathit{G}^j=\mathcal{T}^j+\mathcal{J}\mathit{G}^j$ and using $\|\mathcal{J}\mathit{G}^j\|=\sqrt{R}\,\|\bar{\mathbf{g}}^j\|$ with $\bar{\mathbf{g}}^j=\overline{\nabla f}^{\,j}=\nabla F(\bar{\mathbf{z}}^j)+\Delta^j$ from Lemma~\ref{lem:avg-invariants}, we bound 
\begin{equation}\label{eq:L46}
\begin{aligned}
\|\mathit{G}^j\| & \le \|\mathcal{T}^{j}\| + \sqrt{R}\,\|\bar{\mathbf{g}}^j\|\\ 
& \le\|\mathcal{T}^{j}\| + \sqrt{R}\big(L\|\mathbf{s}^j\| + \tfrac{L}{\sqrt{R}}\|\mathcal{C}^j\|\big),
\end{aligned}
\end{equation}
where we used $L$-smoothness of $F(\mathbf{z})$ and \eqref{eq:delta-bound}. Combining \eqref{eq:L44}--\eqref{eq:L46},
\begin{equation}\label{eq:L47}
\begin{aligned}
\big\Vert\nabla \mathcal{F}^{\,j+1}-\nabla \mathcal{F}^{\,j}\big\Vert
&\le \,L\,c_W\,\|\mathcal{C}^{j}\|
+ \mu\,L^{2}\,\|\mathcal{C}^{j}\|
\\& + \mu\,L\,\|\mathcal{T}^{j}\|
+ \mu\,\sqrt{R}\,L^{2}\,\|\mathbf{s}^{j}\|.
\end{aligned}
\end{equation}
Plugging \eqref{eq:L47} into \eqref{eq:L43} yields, for any $k\in\{i\tau,\ldots,(i+1)\tau\}$,
\begin{equation}\label{eq:L48}
\begin{aligned}
& \Vert\mathcal{T}^{k}\Vert \\
&\le
\epsilon_\tau \Vert\mathcal{T}^{(i-1)\tau}\Vert
+ \sum_{j=(i-1)\tau}^{k-1}\Big(
Lc_W\Vert\mathcal{C}^{\,j}\Vert
+ \mu L^{2}\Vert\mathcal{C}^{\,j}\Vert 
\\& + \mu\,L\,\Vert\mathcal{T}^{\,j}\Vert
+ \mu\,\sqrt{R}\,L^{2}\,\Vert\mathbf{s}^{\,j}\Vert
\Big).
\end{aligned}
\end{equation}
Switching terms, ~\eqref{eq:track-one} then follows.

\end{document}